\documentclass[12pt]{article}
\usepackage{amsmath, amsthm, amssymb, amsfonts, appendix}
\usepackage{natbib}
\usepackage{setspace}
\usepackage{graphicx}
\usepackage{subfig}
\usepackage{placeins}

\usepackage[normalem]{ulem}
\usepackage{color}
\usepackage{xcolor}
\usepackage{afterpage}
\usepackage{float}
\usepackage{tikz}
\usepackage{extdash}
\usepackage[letterpaper, left=1.5in, right=1.5in, top=1.125in, bottom=1.125in, includefoot]{geometry}

\usepackage[T1]{fontenc}
\usepackage[scr = dutchcal]{mathalpha}

\usetikzlibrary{patterns}
\newtheorem{assumption}{Assumption}
\newtheorem{theorem}{Theorem}
\newtheorem*{theorem*}{Theorem}

\newtheorem{lemma}{Lemma}

\theoremstyle{definition}

\theoremstyle{definition}
\newenvironment{proof2}[1][Proof]{\noindent \textit{#1.} }{\  \hfill$\square$\bigskip}
\newcommand \slashfrac[2]{{#1}/{#2}}
\newcommand{\Ind}{\, \rotatebox[origin=c]{90}{\ensuremath{\models}}\,}
\newcommand{\supp}{\text{supp}\,}

\newif \ifshowexplanations

\newif \ifshowflags

\showexplanationsfalse
\showflagsfalse
\begin{document}

%\newgeometry{margin=.7in}
% just for title page, to make it all fit; restoregeometry below undoes this

\begin{titlepage}
\def \footnoterule{ } \title{Nonparametric Identification  of  Differentiated Products Demand Using Micro Data%
\thanks{Early versions of this work were presented in the working paper ``Nonparametric Identification of Multinomial Choice Demand Models with Heterogeneous Consumers,'' first circulated in 2007 and superseded by the present paper. We thank Jesse Shapiro, Suk Joon, Son and numerous seminar participants for helpful comments. Jaewon Lee provided capable research assistance.}  \bigskip
    \bigskip}

\setstretch{1}\author{{\large Steven T. Berry}\\
%EndAName
\textit{Yale University}\\
%\textit{Department of Economics}\\
%\textit{Cowles Foundation}
%\textit{and NBER}
\and
{\large Philip A. Haile}  \\
%EndAName
\textit{Yale University}\\
%\textit{Department of Economics}\\
%\textit{Cowles Foundation, CEMMAP, and NBER}
\bigskip \bigskip}
%\date{\ifcase \month \or January\or February\or March\or April\or May\or June%
%\or July\or August\or September\or October\or November\or December\fi \ %
%\number \day,\  \number \year \\}

%\date{\today}
\date{April 13, 2022} 
\maketitle
\thispagestyle{empty}

\begin{abstract}
  \noindent \setstretch{1}%
  {We examine identification of differentiated products
      demand when one has ``micro data'' linking individual consumers' characteristics and choices. Our model
      nests standard specifications featuring rich observed and unobserved consumer heterogeneity
      as well as product/market-level unobservables that introduce the
      problem of econometric endogeneity. Previous work establishes
      identification of such models using market-level
      data and instruments for all prices and quantities. Micro data provides a panel structure that facilitates richer
      demand specifications and reduces requirements on both the number
      and types of instrumental variables.
      We address identification of demand in the standard case in which
       non-price product characteristics are assumed exogenous, but also cover identification of demand elasticities and other key features when product characteristics are endogenous. We discuss implications of
      these results for applied work.}
\end{abstract}

\end{titlepage}

\pagenumbering{arabic}
\restoregeometry

\section{Introduction}

Systems of demand for differentiated products are central to many
questions in economics. In practice it is common to estimate demand using panel data on the characteristics and choices of many individual consumers within
each market. This is often referred to as ``micro
data,'' in contrast to another common case in which only market-level
outcomes are observed.  At an intuitive level, the panel structure of micro data seems to offer more information than
market-level data alone. But in what precise sense does micro data help? How significant are the advantages of micro data? What specific kinds of variation within and across markets are helpful, and how?

In this paper we consider nonparametric identification of demand, focusing on the particular benefits of micro data. We consider a nonparametric consumer-level demand model
that substantially generalizes parametric models following \cite*{BLP} that
are widely used in practice.  A key benefit of micro data is that unobservables at the level of the product$\times$market
remain fixed as consumer attributes and choices vary within
a given market. This clean ``within'' variation can be combined with  cross-market variation in choice characteristics, market characteristics, prices, and instruments for prices, to yield identification. Compared to settings with market-level data, this can allow both a more general model and substantially reduced demands on instrumental variables.

We focus exclusively on identification.
The celebrated ``credibility revolution'' in applied microeconomics has  redoubled attention to identification obtained
through quasi-experimental variation, such as that arising through instrumental variables, geographic boundaries, or repeated observations within a single economic unit. Identification of demand presents special challenges that are absent in  much of empirical economics (\cite{Berry-Haile-HBK}). Nonetheless, we show that these same types of variation allow  identification of demand systems exhibiting rich consumer heterogeneity and endogeneity. Of course, nonparametric identification results do not eliminate
concerns about the impact of parametric assumptions relied on in
practice. However, they address the important question of whether such
assumptions can be viewed properly as finite-sample approximations
rather than essential maintained hypotheses.  Formal identification
results can also clarify which maintained assumptions may be most
difficult to relax, reveal the essential sources of variation in
the data, offer assurance that robustness analysis is possible,
and potentially lead to new (parametric or nonparametric) estimation
approaches.

Our results also provide insights that can inform applied
practice. Our most important message is that micro data has a high marginal value. Availability of instrumental variables is the most important and challenging requirement for identification of demand, and micro data can substantially reduce both the number and types of instruments needed. With market-level data, nonparametric identification typically requires instruments for all quantities and prices (\cite{Berry-Haile-market,Berry-Haile-annreview,Berry-Haile-HBK}). %
With micro data, we find that the only essential instruments are
those for prices. This cuts the number of required instruments in half and avoids the necessary reliance on so-called ``BLP instruments'' (characteristics of competing products). The availability of micro data also
opens the possible use of classes of instruments that
are often unavailable in the case of market-level data.

Another important finding for applied work concerns the need for cross-market variation: micro data from a single market does not suffice for identification in our model. We discuss the distinct roles of within-market and cross-market variation, including why the latter is needed. An implication is that existing studies using consumer-level data
within just one market may be relying on
functional form restrictions that can escape notice as assumptions of necessity rather than
convenience. These restrictions may include strong
assumptions on how observed consumer
attributes change demand.

We also show that it is sometimes possible to identify the \textit{ceteris
 paribus} effects of prices on quantities demanded (critically, e.g., own- and cross-price demand elasticities) even when observed
non-price product characteristics are endogenous and not instrumented.\footnote{This is related to well-known results regarding endogenous controls in  regression models.} This requires that
instruments for prices remain valid when conditioning on the
endogenous observed product characteristics, and we
illustrate through simple causal graphs how different cases do or do
not satisfy this requirement.  Potentially endogenous product
characteristics are an important concern in the applied literature on
differentiated products, and it can be difficult to find instruments for all such characteristics. Thus, our findings expand the range
of cases in which features of demand of primary interest can be identified despite these
concerns.

Several key aspects of our model and setting are worth emphasizing to clarify our contributions. First, as in the large empirical literature
building on \cite*{Berry94} and \cite*{BLP,MicroBLP}, we emphasize the
role of latent market-level demand shocks (unobservables at the level of the product$\times$market) that result in the
econometric endogeneity of prices. Explicitly accounting for these demand shocks is essential to the identification of policy-relevant features such as demand elasticities and equilibrium counterfactuals.\footnote{See section 2 of \cite{Berry-Haile-HBK} for more extensive discussion of this sometimes underappreciated point.} This drives our focus on market-level
endogeneity and differentiates our work from much of the prior research
on identification of demand or discrete choice models.\footnote{This includes prior work on identification (with micro data) of discrete choice models allowing endogeneity but specifying only a composite ``error'' term for each choice, representing all unobservables (and their interactions with observables) at the level of the product, consumer, or market.  See, e.g., \cite{Lewbel00}.}

Second, the panel structure of consumers-within-markets is essential to the questions we ask. It is what distinguishes micro data from market-level data, and the combination of within- and cross-market variation is essential to the reduction in instrumental variables needs discussed above. This panel structure also allows us to avoid any restriction on the way product-level observables enter demand. These features contrast with those of Berry and Haile (2014), who previously considered identification from market-level data in nonparametric models generalizing that of \cite*{BLP}. In that case, the demand system (especially if combined with supply) closely resembles a nonparametric simultaneous equations model, as studied by, e.g., \cite*{Matzkin_simult_ema,Matzkin_simult_estimate}, \cite*{Benkard-Berry},
\cite*{Blundell-Kristensen-MatzkinP&P,Blundell-Kristensen-Matzkin-individual}, and
\cite*{Berry-Haile-simulteqn}. However, the panel structure essential to the present paper is absent in that prior work.

Third, our model avoids requiring consumer-level observables that can be linked exclusively to the desirability of specific products. Such a requirement (often in combination with large support conditions) is widely used in ``special regressor'' approaches to identification of consumer-level discrete choice models,%
\footnote{%
  See the review by \cite*{Lewbel-special} and references therein. A very early version of this paper, (\cite*{BerryHaileNPold}), featured an example of this
  sort of identification approach in the panel context considered here. One common idea for an exclusive consumer-choice interaction is the distance to each of $J$ choices in geographic space. But even these are inherently restricted to lie on a 2-dimensional surface in $\mathbb{R}_+^J$, since the underlying consumer heterogeneity reflects only consumer locations.} %
  but is often difficult to motivate in practice. 
  More natural are situations in which multiple consumer-level observables interact to alter tastes for  all goods.  As a simple example---one illustrating a broader interpretation of ``demand''---consider a discrete choice model of expressive voting in a two-party (``R'' vs. ``D'') election, applied to survey data matching individual reported 
votes to voter sociodemographics.\footnote{Advertising, rather than price, often plays the role of the endogenous choice
  characteristic---one whose effects are sometimes of primary interest. See, e.g., \cite*{Gerber98} and \cite*{Gordon_voting}.} 
   Although voter-specific measures like age, income, gender, race, and education may provide rich variation in preferences between the two parties (and the outside option to abstain), no such measure is naturally associated exclusively with the attractiveness of a single option.

Fourth, although we initially emphasize discrete choice demand, as in the large applied literature following \cite*{BLP,MicroBLP}, this is not essential. The  primitive feature of interest in our analysis is an ``invertible'' demand function mapping observables (at the level of market, products, and consumer) and  demand shocks (at the level of product$\times$market) to expected quantities demanded.  This can allow continuous demand as well as departures from common assumptions regarding consumers' full information or  rationality.

Of course, our results do require some structure, including conditions on sources of variation. In addition to instruments (for prices) satisfying standard conditions, we rely on three important assumptions. One is a nonparametric
index restriction on the way 
market-level demand shocks and some observed consumer attributes enter the
model.
The second is injectivity of the mappings that link observed consumer attributes to
choice probabilities.\footnote{In contrast to related conditions  in \cite{Berry-Haile-market,Berry-Haile-simulteqn} or \cite{Matzkin_simult_ema}, here each index depends on observed
consumer attributes rather than observed product characteristics (which are fixed within markets), and there is no requirement that these observables be exogenous.} Below we connect these requirements to the literature through some familiar examples.

Finally, we require sufficient variation in the consumer observables
to satisfy a ``common choice probability'' condition that we believe
is new to the literature.
Given values of the non-price observables at the product or market level, this 
condition requires that there exist some point $s^*$ (unknown \textit{a priori}) in the probability
simplex such that in every market one can obtain $s^*$ as the
conditional choice probability vector by conditioning on the ``right''
set of consumer observables for that market. This requires that the number of observed consumer
attributes be at least as large as the number of products, and that they have sufficient independent variation.
However, this requirement contrasts with a standard ``large
support'' condition, which would imply that
\textit{every} point $s$ within the simplex is a common choice
probability. We require only one such point (which may differ with the non-price product/market observables) and allow a broad range of
cases where choice probabilities are never close to one or zero.\footnote{
In the voting example, our condition would require a vote share vector---say 0.4 for R and 0.4 for D (the remainder abstaining)---such that in every market (e.g., metro area) with the same pair of candidates and equal values of any metro-level observables, there is a
  combination of individual-level sociodemographic measures that generates this conditional
  vote share. The level of education etc.\ required to match the given vote
  share might be higher (and perhaps income lower, etc.) in an unobservably conservative market.} An attractive feature of the
common choice probability condition is that is verifiable; i.e., its
satisfaction or failure is identified.

Our results are relevant to a large empirical literature exploiting micro data to estimate demand.  A classic example is McFadden's
  study of transportation demand (\cite*{McFadden-Trans}), where each
  consumer's preferences over different modes of transport are
  affected by her available mode-specific commute times and other factors. This
  example illustrates the defining characteristic of the type of micro
  data considered here: consumer-specific observables that alter the
  relative attractiveness of different options. Consumer distances to
  different options have been used in a number of  applications, including those involving demand for hospitals, retail
  outlets, residential locations, or schools, as in the examples of
  \cite*{CappsDranoveSatterthwaite}, \cite*{BurdaHardingHausman},
  \cite*{Bayer_Keohane_Timmons}, and \cite*{neilson19}. More broadly,
  observable consumer-level attributes that shift tastes for products
  might include a household's income, sociodemographic measures, or
  other proxies for idiosyncratic preferences. For example, income and
  family size have been modeled as shifting preferences for cars
  (\cite*{Goldberg95}, \cite*{Petrin02}); race, education, and birth
  state have been modeled as shifting preferences for residential
  location (\cite*{Diamond_JMP}). Other examples include product-specific advertising exposure
  (\cite*{Ackerberg03}), consumer-newspaper ideological match
  (\cite*{GentzkowShapiro}), and the match between household
  demographics and those of a school or neighborhood
  (\cite*{Bayer_et_al07}, \cite*{HomJMP}).\footnote{%
    Here we cite only a small representative handful of papers out of
    a selection that spans many topics and many years. See also the
    examples in section \ref{sec examples}.} An important feature of
  these examples, reflected by our model, is that the typical
  consumer-level observable cannot be tied exclusively to a single
  good.

In what follows, section \ref{sec model} sets up our model of multinomial choice demand. Section \ref{sec examples} connects this model to parametric
examples from the empirical literature. We 
present our identification results in section \ref{sec identification}. In section \ref{sec discuss} we discuss potential instruments, variations, and extensions,  including models of continuous demand and additional structure that would allow identification when the available micro data or instruments have lower dimension. We discuss other key implications for applied work in section \ref{sec lessons} and conclude in section \ref{sec conclusion}. An appendix provides further discussion of instruments, using simple causal graphs.

\section{Model and Features of Interest}\label{sec model}

\subsection{Setup}

We consider multinomial choice among $J$ goods (or ``products'') and an outside
option (``good 0'') by consumers $i$ in ``markets'' $t$. A market is defined formally by:%
    \footnote{In practice, markets are typically defined by time period or geography. We condition on a fixed number of goods without loss. With additional assumptions,  variation in the number of goods available can be valuable and data from markets with $J$ available goods could be used to predict outcomes in markets with more or fewer goods.}%

\begin{itemize}

\item a price vector $P_{t}=\left( P_{1t},\dots ,P_{Jt}\right)$;

\item a set  of additional observables $X_{t}$;

\item a vector $\Xi _{t}=\left( \Xi _{1t},\dots ,\Xi _{Jt}\right) $ of
unobservables;%
    \footnote{For clarity we write random variables in uppercase and their realizations lowercase. Note that $\Xi $ is the uppercase form of the standard notation $\xi $ for unobservables at the product$\times$market level.}

\item a distribution $F_{YZ}\left( \cdot ;t\right) $ of consumer
observables $(Y_{it},Z_{it}) \in \mathbb{R}^H\times \mathbb{R}^J$, $H\geq 0$, with support $\Omega\left( X_{t}\right) $;

\end{itemize}

The variables $P_t, X_t$, and $\Xi_t$ are common to all consumers in a given market. We distinguish between $P_t$ and $X_t$  due to the particular interest in how demand responds to prices and the typical focus on endogeneity of prices. However, we have not yet made the standard assumption that $X_t$ is exogenous---e.g., independent or mean independent of the demand shocks $\Xi_t$. In fact, we will see below that identification of demand elasticities and other key features of demand can often be obtained without requiring such an assumption (or additional instruments for $X_t$).\footnote{Alternatively, when instruments are available for endogenous components of $X_t$, our results generalize immediately by expanding $P_t$ to include these endogenous characteristics.}

   Although $X_{t}$ will typically include observable product characteristics, it may also include other factors defining markets. For example,  consumers might be partitioned into ``markets'' based on a combination of geography, time, product availability, and average demographics included in $X_{t}$. In contrast, observables varying across consumers within a market are represented by $Y_{it}$ and $Z_{it}$.
 We make a distinction between $Y_{it}$ and $Z_{it}$ in order to isolate our requirements on consumer-level data. Key conditions, made precise below, are that consumer observables have dimension of at least $J$ (hence, $Z_{it}\in \mathbb{R}^J$) and that changes in $Z_{it}$ alter the relative attractiveness of different goods. We do not require the additional consumer observables $Y_{it}$; however, we can accommodate them in an unrestricted way, and conditioning on an appropriate value of $Y_{it}$ can weaken some assumptions. 
 
Although our requirements on $Z_{it}$ permit the case in which each component $Z_{ijt}$ exclusively affects the attractiveness of good $j$, we will not require this. Nor will we require independence (full, conditional, or mean independence) between $(Y_{it},Z_{it})$ and  $\Xi_t$.

The choice environment of consumer $i$ in market $t$ is then represented by
\begin{equation*}
C_{it}=\left( Z_{it},Y_{it},  P_{t},X_{t},\Xi _{t}\right).
\end{equation*}
Let $\mathcal{C}$ denote the support of $C_{it}$. The most basic
primitive characterizing consumer behavior in this setting is a
distribution of decision rules for each $c_{it}\in\mathcal{C}$.%
\footnote{Under additional conditions a distribution
  of decision rules can be represented as the result of utility
  maximization. See, e.g., \cite*{Mas-Colell_Whinston_Green},
  \cite*{BlockMarschak}, \cite*{Falmagne78}, and
  \cite*{McFadden-revealed}. We will not require such conditions 
  or consider a utility-based representation of choice
  behavior. A related issue is identification of welfare effects.
  Standard results allow construction of valid measures of aggregate
  welfare changes from a known demand system in the absence of income
  effects. \cite*{Bhattacharya-WF2} provides such results  for discrete choice settings when
  income effects are present but suggests
  use of control function methods for identification/estimation of
  demand. As discussed in \cite*{Berry-Haile-HBK}, control function  methods are valid only
  under strong functional form restrictions (\cite*{BlundellMatzkin14}), which are violated even
  in standard parametric specifications of demand for differentiated products.}  As usual, heterogeneity in
decision rules (i.e., nondegeneracy of the distribution) within a
given choice environment may reflect a variety of factors, including latent preference heterogeneity
across consumers, shocks to individual preferences, latent variation in consideration sets, or stochastic
elements of choice (e.g., optimization error).\footnote{\label{fn interp}For many purposes, one need not take stand on the interpretation of this randomness, since the economic questions of interest  involve changes to the arguments of demand functions, not to the functions themselves. This covers the canonical motivation for demand estimation: quantifying responses to \textit{ceteris paribus} price changes. However, for some questions---e.g., those involving information interventions or requiring identification of cardinal utilities---the interpretation becomes important.  See \cite*{BCMT-consideration} for a recent contribution on this topic.}

\subsection{Demand and Conditional Demand}

The choice made by consumer $i$ is represented by $Q_{it}=\left( Q_{i1t},\dots ,Q_{iJt}\right)$, where $Q_{ijt}$ denotes the quantity (here, 0 or 1) of good $j$ purchased. Given $C_{it}$, a distribution of decision rules is characterized by the conditional cumulative joint distribution function $F_{Q}\left( q|C_{it}\right) =E\left[ 1\left\{ Q_{it}\leq q\right\} |C_{it}\right]$. In the case of discrete choice, this distribution can be represented without loss by the structural choice probabilities
\begin{equation}\label{eq expected demand}
\mathscr{s} \left( C_{it}\right) = \left( \mathscr{s}_{1}(C_{it}),\dots ,\mathscr{s}_{J}(C_{it})\right) =E\left[ Q_{it}|C_{it}\right].
\end{equation}%
Given the measure of consumers in each choice environment, the mapping $\mathscr{s}$ fully characterizes consumer demand. We therefore consider identification of the demand mapping $\mathscr{s}$ on $\mathcal{C}$.

However, it is useful to also consider identification of the conditional
demand functions%
\[
\bar{\mathscr{s}}\left( Z_{it},Y_{it},P_{t};t\right) =\mathscr{s}\left(
Z_{it},Y_{it},P_{t},x_{t},\xi _{t}\right)
\]%
on
$$\mathcal{C}(x_t,\xi_t)=\supp \left( Z_{it},Y_{it},  P_{t}\right)\vert \left(X_t=x_t,\Xi_t=\xi_t \right)$$
for each market $t$. The function $\bar{\mathscr{s}}\left(Z_{it},Y_{it},P_{t};t\right) $ is simply the demand function $\mathscr{s}$ when $\left( X_{t},\Xi _{t}\right) $ are fixed at the values $\left(
x_{t},\xi _{t}\right) $ realized in market $t$.
Because $\Xi_t$ is unobserved and prices are fixed within each market, identification of $\bar{\mathscr{s}}\left(Z_{it},Y_{it},P_{t};t\right) $ is nontrivial.
However, this mapping fully characterizes the responses of demand (at all combinations of $\left( Z_{it},Y_{it}\right) $) to counterfactual \textit{ceteris paribus}
price variation, holding $X_t$ and $\Xi_t$ fixed at their realized values in market $t$. Thus, knowledge $\bar{\mathscr{s}}(\cdot;t)$ for each market $t$
suffices for many purposes motivating demand estimation in
practice.

Notably, $\bar{\mathscr{s}}(\cdot;t)$ fully determines the own- and
cross-price demand elasticities for all goods in market $t$. One implication
is that  $\bar{\mathscr{s}}(\cdot;t)$ is the feature of $\mathscr{s}$ needed to
discriminate between alternative models of firm competition (e.g., \cite*{Berry-Haile-market},
\cite*{Backus-Conlon-Sinkinson}, \cite*{Duarte-etal-testing}). And,
given an assumed model of supply,  $\bar{\mathscr{s}}(\cdot;t)$ suffices to identify
firm markups and marginal costs, following \cite*{BLP} and \cite*{Berry-Haile-market}; to decompose the
sources of firms' market power, as in \cite*{Nevo2001}; to determine equilibrium outcomes under
a counterfactual tax, tariff, subsidy, or exchange rate (e.g., \cite*{Anderson-Depalma-Kreider}, \cite*{Nakamura-Zerom}, \cite*{Decarolis-etal-subsidy}); or to determine the
equilibrium \textquotedblleft unilateral effects\textquotedblright\ of a
merger (e.g., \cite*{Nevo2000}, \cite*{MillerSheuMerger}).
Furthermore, $\bar{\mathscr{s}}(\cdot;t)$ alone determines the ``diversion ratios'' (e.g., \cite*{Conlon-Mortimer-diversion}) that often play a central role in the
practice of antitrust merger review.

Of course, because the functions  $\bar{\mathscr{s}}(\cdot;t)$  are defined with fixed
values of $\left( X_{t},\Xi _{t}\right) $, they do not suffice for answering
all questions---in particular, those requiring knowledge of \textit{ceteris
paribus} effects of $X_t$ on demand.\footnote{In some cases, such effects may be of direct interest---e.g., to infer
willingness to pay for certain product features. In other cases, such
effects are inputs to determination of demand under counterfactual product
offerings or entry. Thus, while knowledge of $\bar{\mathscr{s}}(\cdot;t)$  in all
markets suffices in a large fraction of applications, knowledge of $%
\mathscr{s}$ is required for others.} However, by avoiding the need to
separate the effects of $X_{t}$ and $\Xi _{t}$ on demand, identification of  $\bar{\mathscr{s}}(\cdot;t)$  in each market $t$ can often be
obtained without requiring exogeneity of $X_{t}$. This can be important when
exogeneity is in doubt and one lacks the additional instruments that
would allow treating endogenous elements of $X_{t}$ as we treat
 prices $P_{t}$ below.
 
 \subsection{Core Assumptions}

So far we have implicitly made two significant assumptions: (i) unobservables at the market level can be represented by a $J$-vector $\Xi_t$; (ii) conditional on  $X_t$, the support of $(Y_{it},Z_{it})$ is the same in all markets. The first is standard but important. The second seems mild for many applications and can be relaxed at the cost of more cumbersome exposition.
Our results will also rely on the following key structure.

\begin{assumption}[Index]\label{ass index}
$\mathscr{s} \left( C_{it}\right) =\sigma \left( \gamma \left(Z_{it},Y_{it},X_{t},\Xi _{t}\right) ,Y_{it},P_{t}, X_{t}\right) $, with \linebreak$\gamma \left(Z_{it},Y_{it},X_{t},\Xi_{t}\right)=\left(\gamma_1 \left(Z_{it},Y_{it},X_{t},\Xi_{t}\right),\dots, \gamma_J \left(Z_{it},Y_{it},X_{t},\Xi_{t}\right)\right) \in \mathbb{R}^{J}.$
\end{assumption}

\begin{assumption}[Invertible Demand]\label{ass invertible demand}
$\sigma \left( \cdot ,Y_{it},P_{t},X_{t}\right) $ is injective on the support of $ \gamma(Z_{it},Y_{it}, X_t,\Xi_t)|(Y_{it},P_t,X_t)$.
\end{assumption}

\begin{assumption}[Injective Index]\label{ass injective index}
$\gamma \left( \cdot ,Y_{it}, X_{t},\Xi _{t}\right) $ is injective on the support of $Z_{it}|(Y_{it},X_t)$.
\end{assumption}

\begin{assumption}[Separable Index]\label{ass linear index}
$\gamma_j \left(Z_{it},Y_{it}, X_t,\Xi _{t}\right) = \Gamma_j\left(
Z_{it},Y_{it},X_t\right) + \Xi_{jt}$ for all $j$.
\end{assumption}

Assumption \ref{ass index} requires that, given $(Y_{it},P_t,X_t)$,
$Z_{it}$ and $\Xi _{t}$ affect choices only through indices
$\left( \gamma _{1}\left( Z_{it},Y_{it},X_{t},\Xi _{t}\right)
  ,\dots,\gamma _{J}\left( Z_{it},Y_{it},X_{t},\Xi _{t}\right) \right)
$ that exclude $P_t$. This is a type of weak separability
assumption. Observe that $X_t$ and $Y_{it}$ can affect demand both
directly and through the indices, and that the indices themselves
enter the function $\sigma$ in fully flexible form. As we illustrate
below, this index structure generalizes standard specifications used
in practice. Assumption \ref{ass invertible demand} further requires
that the choice probability function $\sigma$ be ``invertible'' with
respect to the index vector---that, holding $(Y_{it},P_{t},X_{t})$
fixed, distinct index vectors map to distinct choice probabilities.
This is not without loss, and in general injectivity requires that
$\sigma$ map to interior values, i.e., that $\sigma_j(C_{it})>0$ for
all $j$ and $C_{it}\in \mathcal{C}$.  \cite*{BerryGandhiHaile} provide
sufficient conditions for invertibility and point out that these are
natural in discrete choice settings when each
$\gamma_{j}\left( Z_{it},Y_{it},X_{t}, \Xi _{t}\right) $ can be
interpreted as a (here, consumer-specific) quality index for good $j$.
Assumption \ref{ass injective index} requires injectivity of the index
function $\gamma $ with respect to the vector $Z_{it}$.  This
generalizes common utility-based specifications in which each
$Z_{ijt}$ is assumed to affect only the conditional indirect utility
of good $j$ and to do so monotonically. For example, if each index
were a linear function of the $J$ components of $Z_{it}$, Assumption
\ref{ass injective index} would require the matrix of coefficients to
be full rank.\footnote{General sufficient conditions for injectivity
  can be found in, e.g., \cite*{Palais59}, \cite*{Gale-Nikaido},
  \cite*{Parthasarathy83}, and \cite*{BerryGandhiHaile}.} Assumption
\ref{ass linear index} requires the indices
$ \gamma_{j}\left( Z_{it},Y_{it},X_{t}, \Xi _{t}\right) $ to take an
additively separable structure.\footnote{One key role of additive
  separability here is to allow use of standard IV conditions---those
  needed for identification of separable nonparametric regression
  models---where instruments are required. In the context of
  market-level data, \cite*{Berry-Haile-market} include results that
  allow relaxation of additive separability by strengthening IV
  conditions or other assumptions. See also
  \cite*{Matzkin_simult_estimate} and
  \cite*{Blundell-Kristensen-Matzkin-individual}. None of these covers
  the panel structure (consumers within markets) of the micro data
  setting we consider here.}

\subsection{A Useful Representation of the Index}

Although the formulation of our index vector $\gamma \left(Z_{it},Y_{it},X_{t},\Xi_{t}\right)$ above maximizes clarity about our core assumptions, for the study of identification it will be convenient to define
\begin{equation*}
g\left( Z_{it},Y_{it},X_{t}\right) =\Gamma\left( Z_{it},Y_{it},X_{t}\right) +E
\left[ \Xi _{t}|X_{t}\right]
\end{equation*}
and
\begin{equation}\label{eq define h}
 h\left( X_{t},\Xi _{t}\right) =\Xi _{t}-E\left[ \Xi _{t}|X_{t}\right] ,
\end{equation}%
so that
\begin{equation}\label{eq new index}
 \gamma \left( Z_{it},Y_{it},X_{t},\Xi _{t}\right) =g\left(
 Z_{it},Y_{it},X_{t}\right) +h\left( X_{t},\Xi _{t}\right) .
\end{equation}%
Observe that
\begin{equation}\label{eq location h}
 E\left[ h\left( X_{t},\Xi _{t}\right) |X_{t}\right] =0
\end{equation}%
by construction. 

With this notation, we have
\begin{equation}\label{eq demand}
\mathscr{s}\left( Z_{it},Y_{it},P_{t},X_{t},\Xi _{t}\right) =\sigma \left( g\left(
Z_{it},Y_{it},X_{t}\right) +h\left( X_{t},\Xi _{t}\right)
,Y_{it},P_{t},X_{t}\right)
\end{equation}
and
\begin{equation}\label{eq cond demand}
\bar{\mathscr{s}}\left( Z_{it},Y_{it},P_{t};t\right) =\sigma \left( g\left(
Z_{it},Y_{it},x_{t}\right) +h(x_{t},\xi _{t}),Y_{it},P_{t},x_{t}\right)
\end{equation}
We henceforth work with this representation of the demand and conditional demand functions.

\subsection{Technical Conditions}

Let $\mathcal{X}$ denote the support of $X_t$.  For $x\in \mathcal{X}$, let
 $\mathcal{Y}(x)$ denote
 the support of $Y_{it}|\{X_t=x\}$ and, for $y\in\mathcal{Y}(x)$,  let $\mathcal{Z}(y,x)\subset \mathbb{R}^J$ denote the support of
 $Z_{it}|\{Y_{it}=y,X_t=x\}$. In parts (i)--(iii) of Assumption \ref{ass smoothness} we assume conditions permitting our applications of
 calculus and continuity arguments below. 
 Part (iv)  strengthens the
 injectivity requirements of Assumptions \ref{ass invertible demand}
 and \ref{ass injective index} slightly by requiring that the Jacobian matrices
 $\slashfrac{\partial g(z,y,x)}{\partial z}$ and
 $\slashfrac{\partial \sigma(\gamma,y,p,x)}{\partial \gamma}$ be
 nonsingular almost surely.%
\footnote{\label{fn invariance of domain}Although we state Assumption
   \ref{ass smoothness} with the quantifier ``for all
   $y\in\mathcal{Y}(x)$,'' our arguments require these properties only at
   the arbitrary point $y^0(x)$ selected below for each $x\in \mathcal{X}$. 
   Given parts (i) and (ii), the injectivity of
   $g(\cdot,Y_{it},X_{t})$ required by Assumption \ref{ass injective
     index} implies (by invariance of domain) that the image
   $g\left(\mathcal{O},y,x\right) $ of any open set
   $\mathcal{O}\subseteq\mathcal{Z}(y,x)$ is open.
      An implication is that even without part (iv) there could be no nonempty open set
   $\mathcal{O}\in \mathcal Z(y,x)$ on which
   $\slashfrac{\partial g(z,y,x)}{\partial z}$ was singular, as 
   $g(\mathcal{O},y,x)$ would
   then be a nonempty open subset of $\mathbb{R}^J$, contradicting
   Sard's theorem. A similar observation applies to
   $\slashfrac{\partial \sigma(\gamma,y,p,x)}{\partial \gamma}$.
   }%

\begin{assumption}[Technical Conditions]\label{ass smoothness}
 For all $x\in \mathcal{X}$ and $y\in \mathcal{Y}(x)$,
\newline (i) $\mathcal{Z}(y,x)$ is open and connected;
\newline(ii) $g(z,y,x)$ is uniformly continuous in $z$ on $\mathcal{Z}(y,x)$ and continuously differentiable with respect to $z$ on $\mathcal{Z}(y,x)$;
\newline(iii) $\sigma \left( \gamma,y, p, x \right) $ is continuously differentiable with respect to $\gamma $ for all $\left( \gamma, p \right) \in \supp\left( \gamma \left(Z_{it},Y_{it},X_t,\Xi _{t}\right),P_{t}\right)\vert \{Y_{it}=y,X_t=x\}$;%
and
 \newline(iv) $\slashfrac{\partial g(z,y,x)}{\partial z}$ and  $\slashfrac{\partial \sigma(\gamma,y,p,x)}{\partial \gamma}$ are nonsingular almost surely on $\mathcal{Z}(y,x)$ and $\supp (\gamma(Z_{it},Y_{it},X_t,\Xi_t),Y_{it},P_t,X_t)\vert  \{Y_{it}=y,X_t=x\}$, respectively.
\end{assumption}

\subsection{Normalization \label{subsec normal}}

The model requires two types of normalizations before the
identification question can be properly posed. The first reflects the
fact that the latent demand shocks have no natural location. Thus, we
set $E[\Xi_t]=0$ without loss.  The second reflects the fact that any
injective transformation of the index vector
$\gamma \left(Z_{it},Y_{it},X_t,\Xi _{t}\right) $ can be reversed by
appropriate modification of the function $\sigma$.  For example, take
arbitrary $A(X_t):\mathcal{X}\to\mathbb{R}^J$ and
$B(X_t):\mathcal{X}\to \mathbb{R}^{J\times J}$ ($B(x)$
invertible at all $x$). By letting { \footnotesize
\begin{align*}
\tilde{\gamma}\left( Z_{it},Y_{it},X_t,\Xi_t \right)  &=A(X_t)+B(X_t)\gamma \left( Z_{it},Y_{it},X_t,\Xi_t
\right)  \\
\tilde{\sigma}\left( \tilde{\gamma}\left( Z_{it},Y_{it},X_t,\Xi_t \right) ,Y_{it},P_t, X_t\right)
&=\sigma \left( B(X_t)^{-1}\left( \tilde{\gamma}\left( Z_{it},Y_{it},X_t,\Xi_t \right)
-A(X_t)\right), Y_{it},P_t,X_t\right)
\end{align*}%
}%
one obtains a new representation of the same distribution of
decision rules (and thus same demand), the new one satisfying our assumptions whenever the original does. We must choose a single representation of demand before exploring whether the observables allow identification.\footnote{Like location and scale normalizations of utility functions, our normalizations place no restriction on the demand function $\mathscr{s}$ or the conditional demand functions $\bar{\mathscr{s}}(\cdot;t)$. \label{fn ambiguity}However, our example illustrates an inherent ambiguity in the interpretation of how\ a given variable alters preferences. For example, in terms of consumer behavior (e.g., demand), there is no difference between a change in $Z_{ijt}$ (all else fixed) that makes good $j$ more desirable and a change in $Z_{ijt}$ that makes all other goods (including the outside good) less desirable. In practice, this ambiguity is often resolved with \textit{a~priori} exclusion assumptions---e.g., an assumption that $Z_{ijt}$ affects only the utility obtained from good $j$.  Such assumptions could only aid identification. See, for example, section \ref{sec stronger index} below.}

To do this, for each $x$ we take an arbitrary
$\left( z^{0}(x),y^{0}(x)\right) $ from the support of
$\left( Z_{it},Y_{it}\right) |\left\{ X_{t}=x\right\} $.  We then
select the representation of demand in which
\begin{equation}
E\left[ \gamma \left( z^{0}\left( X_t\right) ,y^{0}\left( X_t\right) ,X_t,\Xi_t
\right) |X_t=x\right] =0\text{ }\forall x  \label{eq location gamma}
\end{equation}%
and
\begin{equation}\label{eq normalize g matrix}
\left[ \frac{\partial g\left( z^{0}\left( x\right) ,y^{0}\left( x\right)
,x\right) }{\partial z}\right] = I\text{ }\forall x,
\end{equation}
where $I$ denotes the $J$-dimensional identity matrix. Observe that (\ref{eq new index}), (\ref{eq location h}), and (\ref%
{eq location gamma}) together imply%
\begin{equation}\label{eq location g}
g\left( z^{0}\left( x\right) ,y^{0}\left( x\right) ,x\right) =0 \quad \forall x.
\end{equation}

In the example above this choice of normalization is equivalent to taking
\[
B\left( x\right) =\left[ \frac{\partial g\left( z^{0}\left( x\right)
,y^{0}\left( x\right) ,x\right) }{\partial z}\right] ^{-1}
\]%
and
\[
A(x)=-B(x)g\left(z^0(x),y^0(x), x \right)
\]%
at each $x$, then dropping the tildes from the transformed model.

\section{Parametric Examples from the Literature}\label{sec examples}

The empirical literature includes many examples of parametric
specifications that are special cases of our model. Discrete choice demand models
are frequently formulated using a random coefficients random utility specification such as
\begin{equation}
u_{ijt}=x_{jt}\beta _{ijt}-\alpha_{it} p_{jt} +\xi _{jt}+\epsilon _{ijt},  \label{eq RCRUM}
\end{equation}%
where $u_{ijt}$ represents individual $i$'s conditional indirect utility
from choice $j$ in market $t$. 
The additive $\epsilon _{ijt}$ is typically specified as a draw from a type-1 extreme value  or
normal distribution, yielding a mixed multinomial logit or probit model. Components $k$ of the random coefficient vector $\beta
_{ijt}$ are often specified as
\begin{equation}\label{eq betai}
    \beta _{ijt}^{(k)} = \beta_{0j}^{(k)}+\sum_{\ell =1}^{L}\beta_{zj}^{(k,\ell)}z_{i\ell t} + \beta_{\nu j}^{(k)}\nu _{it}^{(k)},
\end{equation}
where each $z_{i\ell t}$ represents an observable characteristic of
individual $i$, and each $\nu _{it}^{(k)}$ is a random variable with a
pre-specified distribution. Often, the coefficient on price
is also specified as varying with some observed consumer
characteristics $y_{it}$, such as income. A typical specification of $\alpha_{it}$ takes the form
\begin{equation}\label{eq alphai}
\ln(\alpha_{it}) = \alpha_0 + \alpha_{y} y_{it} + \alpha_{\nu}
\nu^{(0)}_{it}.
\end{equation}

With (\ref{eq betai}) and (\ref{eq alphai}),
we can rewrite $\left( \ref{eq RCRUM}\right) $ as
\begin{equation}
u_{ijt}=g_{j}\left( z_{it},x_t\right) +\xi _{jt}
  +\mu _{ijt},
\label{eq ARUM rewrite mu}
\end{equation}%
where
\[
g_{j}\left( z_{it},x_t\right) =\sum_{k}x_{jt}^{(k)}\sum_{\ell =1}^{L}\beta_{zj}^{(k,\ell)}z_{i\ell t}
\]
\begin{equation}\label{eq def mu}
  \mu _{ijt}= \sum_{k}x_{jt}^{(k)} \left(\beta^{(k)}_{0j}+\beta_{\nu j}^{(k)}\nu
_{it}^{(k)}\right)-p_{jt}\exp(\alpha_0 +\alpha_y y_{it}+\alpha_{\nu}\nu^{(0)}_{it})+\epsilon_{ijt}.
\end{equation}
Observe that all effects of $z_{it}$ and $\xi _{t}$ operate though
indices
\[
\gamma _{j}\left( z_{it},x_t,\xi _{t}\right) = g_{j}\left( z_{it},x_t\right)
+ \xi _{jt}\qquad
j=1,\dots ,J,
\]%
satisfying our Assumptions 1 and 4. It is easy to show that the
resulting choice probabilities satisfy Berry, Gandhi and Haile's (2013)
``connected substitutes'' condition with respect to the vector of indices
$\left( \gamma _{1}\left(
    z_{it},x_t,\xi _{t}\right) ,\dots ,\gamma _{J}\left( z_{it},x_t,\xi
    _{t}\right) \right) $; therefore, the injectivity of demand
required by Assumption 2 holds. Our assumptions require $L\geq J$.%
\footnote{If $L>J$, we can combine the ``extra'' components of $Z_{it}$ with income to redefine the partition of consumer observables as $(Y_{it},Z_{it})$ with $Z_{it}\in \mathbb{R}^J$. More generally, income and any extra components of $Z_{it}$ may affect both the index  (reintroducing $Y_{it}$ as an argument of $g$) and the coefficients on $(X_t,P_t)$.}  Injectivity of
$g(z_{it},x_t) = (g _{1}( z_{it},x_t) ,\dots ,g_{J}(z_{it},x_t))$
in $z_{it}$ (Assumption \ref{ass injective index}) might then be assumed as a primitive condition of the model or
 derived from other conditions, as in the example
we discuss below.\footnote{Although we discuss only the core assumptions, the technical conditions of Assumption 5 can be confirmed for these examples as long as $\supp Z_{it}|\{Y_{it},X_t\}$ is open and connected.}

Of course, our model does not rely on the linear structure of this
example, on any parametric distributional assumptions, or on a
representation of demand through random utility maximization.  But
this example connects our model to a large number of applications and
shows one way that the individual-level observables $z_{it}$ can
interact with product characteristics to generate preference
heterogeneity across consumers facing the same choice set (i.e., where
all $x_{jt},p_{jt}$ and $\xi _{jt}$ are fixed). Note that this
standard specification lacks features sometimes relied on in results
showing identification of discrete choice models: in addition to the
absence of individual characteristics that exclusively affect the
utility from one choice $j$, this model does not exhibit independence
between the ``error term'' $(\xi_{jt}+ \mu _{ijt})$ in (\ref{eq ARUM
  rewrite mu}) and any of the observables $z_{it}, x_t,p_t$.%
\footnote{This is true even without the demand shocks $\xi_{jt}$. For
  example, the individual ``taste shock'' vectors $\nu_{it}$ and
  $\epsilon_{it}$ are typically assumed independent across $i$ and
  $t$; however, $x_{jt}$ and $p_{jt}$ enter the composite error
  $\mu_{ijt}$. Likewise, $z_{it}$ may be correlated with $y_{it}$,
  which enters $\mu_{ijt}$. Furthermore, $x_{jt}$ and $p_{jt}$ may be
  correlated with changes in the distribution of $z_{it}$ across
  markets, introducing variation in this distribution with
  $\mu_{ijt}$}. %

\nocite{Ho_jobmkt}
To see another way that our index structure arises in practice,
consider Ho's (2009)  model of demand for health insurance.
Each consumer $i$ in market $t$ considers $J$ insurance plans as well
as the outside option of remaining uninsured. Each consumer has
a vector of observable characteristics $d_{it}$ (used below to define $z_{it})$.\footnote{Ho's data include measures of individual age, gender, income, home location, employment status, and industry of employment.}
Let $n_{jt}$ denote the set of hospitals in plan $j $'s
network, along with their characteristics (e.g., location and the
availability of speciality services like cardiac care). Each insurance
plan is associated with its network $n_{jt}$, an annual premium
$p_{jt}$, additional observed plan characteristics $x_{jt}$ (e.g., the
size of its physician network), and an unobservable $\xi _{jt}$.%

A consumer's insurance plan demand depends on her particular
likelihood of having of each type of hospital need (diagnosis), as
well as how her preferences over hospital characteristics will vary
with the type of need.  This gives each consumer $i$ an expected
utility $EU\left( n_{jt},d_{it}\right) $ for the option to use plan
$j$'s hospital network. Ho derives this expected utility from
auxiliary data on hospital choice (see also \cite*{Ho_demand} and
\cite*{Ho_Lee_Insurer} ).  From the perspective of identification,
this yields a known functional form for the consumer-specific measures
\[
z_{ijt}\equiv EU\left( n_{jt},d_{it}\right) 
\]%
entering consumer $i$'s conditional indirect utilities 
\begin{equation}
  \label{eq Ho logit}
u_{ijt}=\lambda z_{ijt}+x_{jt}\beta -\alpha(y_{it}) p_{jt}+\xi
_{jt}+\epsilon _{ijt}
\end{equation}%
for each plan $j$. Here $y_{it}\in d_{it}$ represents the consumer's income. Ho assumes each $\epsilon_{ijt}$ is an independent draw from a type-1 extreme value distribution, yielding a multinomial logit model.

Observe that in this example Ho combines data on the characteristics of consumers and choices with additional modeling to derive a scalar $z_{ijt}$ that
exclusively affects only the utility of choice $j$.\footnote{Our model would also allow the possibility that the function $EU$ here is not learned from auxiliary data; in that case our $g_j(z_{it},y_{it},x_t)$ would play the role of $EU\left(n_{jt},d_{it}\right)$, with $z_{it}=d_{it}$ and $n_{jt}\subset x_{jt}$.}   In this case, the injectivity of the index vector
$\gamma \left( z_{it},\xi _{t}\right)$ required by our Assumption 3 holds as long as $\lambda \neq 0$.   Satisfaction of our remaining assumptions follows as in the previous example.

\section{Identification}\label{sec identification}

We consider identification of the demand system
$$
\mathscr{s}(Z_{it},Y_{it},P_t,X_t,\Xi_t)
$$
and the conditional demand systems
$$
\bar{\mathscr{s}}(Z_{it},Y_{it},P_t;t)
$$
from observation of the choice decisions of the population of consumers $i$ in a population of markets $t$. The observables comprise $Z_{it},Y_{it}, P_{t},X_{t},Q_{it},$ and a vector of instruments $W_t$ discussed below. These observables imply observability of choice probabilities conditional on $(Z_{it},Y_{it},P_t,X_t)$ in each market $t$.

Because our arguments do not require variation in $Y_{it}$, in much of what follows we will fix $Y_{it}$ (conditional on $X_t$) at $y^0(X_{it})$. We proceed in three steps. First, in section \ref{sec id index} we present lemmas demonstrating identification of the function $g(\cdot,y^0(x),x)$ at each $x\in \mathcal{X}$. Second, in section \ref{sec conditional demand} we use this result to link latent market-level variation in $h(X_t,\Xi_t)$ to variation in the observed value of $Z_{it}$ required to produce a given conditional choice probability in each market. In particular, given instruments for prices, we show  that the realized values $h(x_t,\xi_t)$ can be pinned down in every market, making identification of the conditional demand systems $\bar{\mathscr{s}}(\cdot;t)$ in each market straightforward.  Finally, in section \ref{sec id demand} we show that $\mathscr{s}$ is also identified when one adds the usual assumption that $X_t$ is exogenous. Thus, after the initial setup and lemmas, the main results themselves follow relatively easily. 

Before proceeding, we provide some key definitions and observations.  For  $\left( p,x,\xi\right) \in \supp\left(P_{t}, X_t,\Xi _{t}\right) $ let
\begin{equation}\label{eq script S}
\mathcal{S}\left( p,x,\xi \right) =\sigma \left( g\left( \mathcal{Z}(y^0(x),x),y^0(x),x\right)+h(x,\xi) ,y^0(x),p,x\right).
\end{equation}%
Thus, $\mathcal{S}\left( p,x,\xi\right)$ denotes the support of choice
probabilities in any market $t$ for which $P_t=p,X_t=x$, and
$\Xi_{t}=\xi$ (holding $Y_{it}=y^0(x)$).%
\footnote{Because
  $\mathcal{Z}(y^0(x),x)$ is open, continuity and injectivity of
  $\sigma$ with respect to the index and of the index with respect to
  $Z_{it}$ imply (by invariance of domain) that
  $\mathcal{S}\left( p,x,\xi\right) $ is open.}
  By Assumptions \ref{ass invertible demand} and \ref{ass injective
    index}, for each $s\in\mathcal{S}\left(x,p, \xi \right) $ there
  must be a unique $z^{\ast }\in \mathcal{Z}(y^0(x),x)$ such that
  $\sigma \left( g\left( z^{\ast },y^0(x),x\right)
    +h(x,\xi),y^0(x),p,x\right) =s$. So for
  $\left( p,x,\xi\right) \in \supp\left(P_{t}, X_t,\Xi_{t}\right)$ and
  $s\in \mathcal{S}\left(p,x, \xi \right)$, we define the function
$$z^{\ast }\left(s;p,x,\xi\right) $$
implicitly by
\begin{equation}\label{eq define z star}
    \sigma \left( g\left( z^{\ast }\left(s;p,x,\xi\right),y^0(x),x \right) +h(x,\xi),y^0(x),p,x\right) =s.
\end{equation}

This definition leads to two observations that play key roles in what follows. First,  in each market $t$ the set $\mathcal{S}(p_t,x_t,\xi_t)$ and the values of $z^{\ast }\left( s;p_{t},x_t,\xi_{t}\right) $ for all $s\in\mathcal{S}\left( p_{t},x_t,\xi _{t}\right) $ are  observed, even though the value of the argument $\xi _{t}$ is not.  Second, by the invertibility of $\sigma$ (Assumption \ref{ass invertible demand}),  we have
\begin{equation}\label{eq inverted demand}
g\left( z^{\ast }\left( s;p,x,\xi\right),y^0(x),x \right) +h(x,\xi) =\sigma ^{-1}\left(s;y^0(x),p,x\right)
\end{equation}%
for all $\left(p,x, \xi \right) \in \supp\left(P_{t},X_t,\Xi_{t}\right)$ and $s\in \mathcal{S}\left( p,x,\xi\right)$.

\subsection{Key Lemmas}\label{sec id index}

Let $||\cdot ||$ denote the Euclidean norm.  We will require the following nondegeneracy condition.%

\begin{assumption}[Nondegeneracy]\label{ass support xi}
For each $x\in\mathcal{X}$, there exists  $p\in \supp P_{t}|\{X_t=x\}$ such that  $\supp \Xi _{t}|\{P_{t}=p,X_t=x\}$ contains an open subset of $\mathbb{R}^J$.
\end{assumption}

Assumption \ref{ass support xi} requires continuously distributed $\Xi_t$  but is otherwise mild. It rules out trivial cases in which conditioning on  $(P_t,X_t)$ indirectly fixes $\Xi_t$ as well. This nondegeneracy condition is implied by standard models of supply, where prices respond to continuous cost shifters or markup shifters (observed or unobserved), allowing the same equilibrium price vector $p$ to arise under different realizations of $\Xi_{t}$.
A key implication, exploited to prove Lemma \ref{lem open ball} below, follows from the definition (\ref{eq define h}): for each $x\in\mathcal{X}$ there exist $\epsilon >0$ and $p\in \supp P_{t}|\{X_t=x\}$ such that for any $d\in\mathbb{R}^{J}$ satisfying $||d||<\epsilon $, $\supp\Xi_t |\{P_t=p,X_t=x\}$ contains vectors $\xi $ and $\xi ^{\prime }$ satisfying $h(x,\xi) -h(x,\xi ^{\prime})=d.$

\begin{lemma}\label{lem open ball}
Let Assumptions \ref{ass index}--\ref{ass support xi} hold. For each $x\in\mathcal{X}$, there exist $p\in \supp P_{t}|\{X_t=x\}$ and 
$\Delta>0$ 
such that
for all  $z$ and $z^{\prime }$ in $\mathcal{Z}(y^0(x),x)$ satisfying
$\left\vert \left\vert  z^{\prime } - z\right\vert \right\vert <\Delta$, 
there exist a choice probability vector $s$ and vectors  $\xi$ and  $\xi ^{\prime }$ in $\supp\Xi _{t}|\{P_t=p,X_t=x\}$ such that $z=z^{\ast }\left( s;p,x,\xi\right)$ and $z^{\prime}=z^{\ast }\left( s;p,x,\xi ^{\prime }\right)$. Furthermore, such $(\Delta,p)$ are identified.
\end{lemma}

\begin{proof}
Fix a value of $x\in\mathcal{X}$. By Assumption \ref{ass support xi}, there exist
$p\in \supp P_{t}|\{X_t=x\}$ 
and $\epsilon >0$ 
such that for any  $z$ and $z^{\prime }$ in $\mathcal{Z}(y^0(x),x)$ for which
 \begin{equation}\label{eq gdif lt epsilon}
 \left\vert \left\vert g\left( z^{\prime },y^0(x),x\right) - g\left(z,y^0(x),x\right) \right\vert \right\vert <\epsilon,
 \end{equation}
 there exist $\xi $ and $\xi^{\prime }$ in  $\supp\Xi_{t}|\{P_{t}=p,X_t=x\}$ such that 
 $$h(x,\xi) -h(x,\xi^{\prime}) =g\left(z^{\prime},y^0(x),x\right) -g\left( z,y^0(x),x\right),$$ 
 i.e., $\gamma(z',y^0(x),x,\xi')=\gamma(z,y^0(x),x,\xi).$
Taking $$s=\sigma \left( \gamma(z',y^0(x),x,\xi') ,y^0(x),p,x\right) =\sigma \left(\gamma(z,y^0(x),x,\xi),y^0(x),p,x\right),$$ the definition (\ref{eq define z star}) implies that
\begin{equation}\label{eq two zs}
    z=z^{\ast }\left( s;p,x,\xi\right) \quad \text{and} \quad z^{\prime}=z^{\ast }\left( s;p,x,\xi ^{\prime }\right).
\end{equation}  
By uniform continuity of $g(\cdot, y^0(x),x)$, there exists $\Delta>0$ such that (\ref{eq gdif lt epsilon}) holds whenever 
\begin{equation}\label{eq zdiff lt Delta}
    \left\vert\left\vert  z^{\prime } - z\right\vert \right\vert <\Delta.
\end{equation}
Because  satisfaction of (\ref{eq zdiff lt Delta}) and (\ref{eq two zs}) 
is observable, all $(\Delta,p)$ allowing satisfaction of these conditions are identified. 
\end{proof}

With this result in hand, we can use equation (\ref{eq inverted demand}) to relate partial derivatives of $g(z,y(x),x)$ at any point $z$  to those at nearby points $z^{\prime}$ by examining the change in consumer characteristics required to create a given change in the vector of choice probabilities.

\begin{lemma}\label{lem g id in ball}
Let Assumptions \ref{ass index}--\ref{ass support xi} hold. Then for every $x\in \mathcal{X}$ there exists a known $\Delta >0$ such that for almost all $z,z^{\prime }\in \mathcal{Z}(y^0(x),x)$ satisfying (\ref{eq zdiff lt Delta}) the matrix $\left[\frac{\partial g(z,y^0(x),x)}{\partial z}\right] ^{-1}\left[ \frac{\partial g(z^{\prime },y^0(x),x)}{\partial z}\right]$ is identified.
\end{lemma}

\begin{proof2}
Given any $x\in \mathcal{X}$,   take a (known) $(p,\Delta)$ as in Lemma \ref{lem open ball}. Consider markets $t$ and $t^{\prime }$ in which $P_{t}=P_{t^{\prime }}=p$ but, for some choice probability vector $s$,
\begin{equation}\label{eq two z}
z=z^{\ast }\left( s;p,x,\xi _{t}\right) \neq z^{\prime }=z^{\ast }\left(s;p,x,\xi _{t^{\prime }}\right),
\end{equation}
revealing that $\xi_{t}\neq \xi _{t^{\prime }}$.
Lemma \ref{lem open ball} ensures that such $t,t',$ and $s$ exist for all $z,z^{\prime }\in \mathcal{Z}(y^0(x),x)$ satisfying  (\ref{eq zdiff lt Delta}).
And although $\xi_t$ and $\xi_{t^{\prime}}$ are latent, the identities
of markets $t$ and $t'$ satisfying (\ref{eq two z}) are observed, as
are the associated values of $s$, $z^*(s;p,x,\xi_t)$, and
$z^*(s;p,x,\xi_{t'})$.  Differentiating (\ref{eq inverted demand})
with respect to the vector $s$ within these two markets,%
we obtain
\begin{equation*}
\frac{\partial g\left( z,y^0(x),x\right) }{\partial z}\frac{\partial z^{\ast }\left(s;p,x,\xi _{t}\right) }{\partial s} = \frac{\partial\sigma^{-1}\left(s;y^0(x),p,x\right) }{\partial s}
\end{equation*}%
and%
\begin{equation*}  %\label{eq partial gzs 2}
\frac{\partial g\left( z^{\prime },y^0(x),x\right) }{\partial z}\frac{\partial z^{\ast }\left( s;p,x,\xi _{t^{\prime }}\right) }{\partial s} = \frac{\partial\sigma^{-1}\left( s;y^0(x),p,x\right) }{\partial s}.
\end{equation*}%
Thus, recalling Assumption \ref{ass smoothness}, for almost all such $(z,z')$ we have
\begin{equation*}  %\label{eq partial gzs 3}
\left[ \frac{\partial g\left( z^{\prime },y^0(x),x\right) }{\partial z}\right] ^{-1}\frac{\partial g\left( z,y^0(x),x\right) }{\partial z} =
\frac{\partial z^{\ast } \left(s;p,x,\xi_{t^{\prime}}\right) }{\partial s} \left[ \frac{\partial z^{\ast }\left( s;p,x,\xi _{t}\right) } {\partial s}\right] ^{-1}.
\end{equation*}%
The matrices on the right-hand side are observed.
\end{proof2}

This leads us to the main result of this section, obtained by
connecting (for each value of $x$) the matrix products
$\left[\frac{\partial g(z,y^0(x),x)}{\partial z}\right] ^{-1}\left[
  \frac{\partial g(z^{\prime },y^0(x),x)}{\partial z}\right]$
identified in Lemma \ref{lem g id in ball} to the known (normalized)
value of the matrix
$\left[ \frac{\partial g(z,y^0(x),x)}{\partial z}\right]$ at
$z=z^0(x)$.

\begin{lemma}\label{thm g identified}
Under Assumptions \ref{ass index}--\ref{ass support xi}, $g(\cdot,y^0(x),x)$ is identified on $\mathcal{Z}(y^0(x),x)$  for all $x\in\mathcal{X}$.
\end{lemma}

\begin{proof2}
For $\epsilon>0$, let $\mathcal{B}\left(b,\epsilon \right) $ denote an open ball in $\mathbb{R}^{J}$ of radius $\epsilon$, centered at $b.$ Take any $x\in\mathcal{X}$ and associated $\Delta>0$ as in Lemma \ref{lem g id in ball}. For each vector of integers $\tau \in \mathbb{Z}^{J}$, define the set
\begin{equation*}
\mathcal{B}_{\tau }=\mathcal{Z}(y^0(x),x)\, \cap\, \mathcal{B}\left( z^{0}(x) + \frac{\tau \Delta}{J}, \frac{\Delta}{2} \right).
\end{equation*}%
By construction, all $z$ and $z^{\prime }$ in any given set $\mathcal{B}_{\tau }$ satisfy (\ref{eq zdiff lt Delta}). So by Lemma \ref{lem g id in ball}, the value of
\begin{equation}\label{eq ratio}
\left[ \slashfrac{\partial g(z,y^0(x),x)}{\partial z}\right]^{-1} \left[\slashfrac{\partial g(z^{\prime },y^0(x),x)}{\partial z}\right]
\end{equation}
is known for almost all $z$ and $z^{\prime }$ in any set
$\mathcal{B}_{\tau }$.  Because
$\cup_{\tau \in \mathbb{Z}^{J}}\mathcal{B}_{\tau }$ forms an open
cover of $ \mathcal{Z}(y^0(x),x)$, %
given any $z\in \mathcal{Z}(y^0(x),x)$ there exists a simple chain of
open sets $\mathcal{B}_{\tau }$ in $\mathcal{Z}(y^0(x),x)$ linking the
point $z^0(x)$ to $z$.%
\footnote{See, e.g., van Mill (2002, Lemma 1.5.21).}  Thus,
$$\left[ \slashfrac{\partial g(z,y^0(x),x)}{\partial z}\right]^{-1} \left[\slashfrac{\partial g(z^0(x),y^0(x),x)}{\partial z}\right]
$$
is known for almost all $z\in\mathcal{Z}(y^0(x),x)$. With the normalization (\ref{eq normalize g matrix}) and the continuity of $\slashfrac{\partial g(z,y^0(x),x)}{\partial z}$ with respect to $z$,%
the result then follows from the fundamental theorem of calculus for line integrals and the boundary condition (\ref{eq location g}).%
\end{proof2}

\nocite{van-Mill-topology}

Before moving to identification of conditional demand, we pause to
point out that our constructive identification of $g(\cdot,y^0(x),x)$ used only
a single price vector $p$ at each value of $x$---that required by Assumption \ref{ass
  support xi}. In typical models of supply this condition would hold
for almost all price vectors in the support of $P_t|\{X_t=x\}$.  In addition to
providing falsifiable restrictions, this indicates a form of
redundancy that would typically be exploited by estimators used in
practice. Similarly, our proof of Lemma \ref{thm g identified} used,
for each $z\in \mathcal {Z}(y^0(x),x)$, only one of infinitely many paths
between $z^0$ and $z$; integrating along any such path must yield the
same function $g(\cdot,y^0(x),x)$ at each $x$.

\subsection{Identification of Conditional Demand}\label{sec conditional demand}

We demonstrate identification of the conditional demand functions $\bar{\mathscr{s}}(\cdot;t)$  under two additional conditions.  The first is a requirement of sufficient variation in the consumer-level observables $Z_{it}$.

\begin{assumption}[Common Choice Probability]\label{ass CCP}
For each $x \in \mathcal{X}$, there exists a choice probability vector $s^{\ast }(x)$ such that $s^{\ast }(x)\in \mathcal{S}\left( p,x,\xi \right) $ for all $\left(p, \xi \right) \in \supp \left( P_{t},\Xi _{t}\right)|\{X_t=x\}$.
\end{assumption}

Assumption \ref{ass CCP} requires that, at each $x\in \mathcal{X}$,  there exist some choice
probability vector $s^{\ast }(x)$ that is common to all markets---that
\begin{equation*}
  \bigcap_{\left( p,\xi \right) \,\in\, \supp \left( P_{t},\Xi _{t}\right)|\{X_t=x\}
}\mathcal{S}\left( p,x,\xi \right)
\end{equation*} 
be nonempty. The nondegeneracy of
each set $\mathcal{S}\left( p_t,x_t,\xi_t\right) $ (recall (\ref{eq
  script S})) reflects variation in $Z_{it}$ across its
support. Assumption \ref{ass CCP} requires enough variation in
$Z_{it}$ that for some $s^{\ast }(x)$ we have
$s^{\ast }(x)\in \mathcal{S}\left( p_{t},x,\xi _{t}\right) $ for all
$\left( p_{t},\xi_{t}\right) $ in their support conditional on $X_t=x$. 

The strength of this assumption
depends on the joint support of $\left( P_{t},\Xi _{t}\right)$ given $\{X_t=x\}$  and on
the relative impacts of $\left( Z_{it},\Xi _{t},P_{t}\right) $ on
choice behavior. Observe that $P_{jt}$ and $\Xi _{jt}$ typically will
have opposing impacts and will be positively dependent conditional on $X_t$ under
equilibrium pricing behavior; thus, large support for
$g\left( Z_{it},y^0(x),x\right) $ may not be required even if $\Xi _{t}$ were
to have large support. Indeed, we can contrast our assumption with a requirement of 
special regressors with large support: the latter would imply that
\textit{every} interior choice probability vector $s$ is a common
choice probability for all $x$; we require only a single common choice probability at each $x$.  Note also that, because choice probabilities
conditional on $(Z_{it},Y_{it})$ are observable in all markets, Assumption~\ref{ass
  CCP} is verifiable.%
\footnote{See \cite*{Berry-Haile-simulteqn} for a formal definition of
  verifiability.} This is important on its
  own. And, because the choice of each $y^0(x)$ is arbitrary,
it implies that we require only existence (for each $x$) of one such
$y^0(x) \in \mathcal{Y}$ such that Assumption \ref{ass CCP}
holds.\footnote{When more than one such value $y^0(x)$ exists, or when
  there is more than one common choice probability vector $s^*$, this
  introduces additional falsifiable restrictions.} Finally, an important observation for what follows is that
the values of any common choice probability vectors $s^*(x)$ may be treated
as known.

Our second requirement is existence of instruments for prices satisfying the standard nonparametric IV conditions.

\begin{assumption}[Instruments for Prices]\label{ass exclusion}\label{ass completeness}\label{ass IV} \phantom{b} \newline
(i) $E\left[ h_j(X_t,\Xi _{jt})|X_t,W_{t}\right] = E\left[ h_j(X_t,\Xi _{jt})|X_t\right]$ almost surely  for all $j=1,\dots ,J$;\newline
(ii) In the class of functions $\Psi\left( X_t, P_{t}\right) $ with finite expectation,\newline $E\left[\Psi\left( X_t,P_{t}\right) |X_t,W_{t}\right] =0$ almost surely implies $\Psi\left( X_t, P_{t}\right) =0$ almost surely.
 \end{assumption}

Part (i) of
Assumption \ref{ass IV} is the exclusion
restriction, requiring that variation in $W_{t}$ not alter the mean of
the latent $h(X_t,\Xi_t)$ conditional on $X_t$.  Recall that    $E[h(X_t,\Xi_t)|X_t]=0$ by construction; thus part (i) implies
\begin{equation}\label{eq NP IV}
    E\left[ h_j(X_t,\Xi _{jt})|X_t,W_{t}\right]=0 \quad \text{a.s.\ for all $j$.}
\end{equation}
This is true regardless of whether $X_t$ itself is exogenous. Of course, one must be cautious about satisfaction of part (i) when $X_t$ is thought to be endogenous. In general, candidate instruments that are properly excludable unconditionally may not be so conditional on an endogenous control. We discuss this further below and devote the appendix to a detailed discussion of when  standard instruments for prices will (or will not)  satisfy the exclusion requirement when $X_t$ is endogenous.  Part (ii) is a standard completeness
condition---the nonparametric analog of the classic rank condition for
linear regression. For example, \cite*{NeweyPowell2003} have shown
that under mean independence (the analog of (\ref{eq NP IV}) here), completeness is necessary and sufficient
for identification in separable nonparametric regression. The
following result demonstrates that, given existence of a common choice probability vector
$s^*$, the same instrumental variables conditions suffice here to
allow identification of $h_j(x_t,\xi_{jt})$ for all $j$ and $t$.

\begin{lemma}\label{lem residuals identified}
Under Assumptions \ref{ass index}--\ref{ass completeness}, the scalar $h_j(x_t,\xi_{jt})$ is identified for all $j$ and $t$.
\end{lemma}

\begin{proof2}
Taking $x=x_t, p=p_t, \xi=\xi_t$ and $s=s^{\ast }(x_t)$ in equation (\ref{eq inverted demand}), we have
$$g\left( z^{\ast }\left( s^{\ast }(x_t);p_{t},x_t,\xi _{t}\right),y^0(x_t),x_t \right) =\sigma^{-1}\left( s^{\ast }(x_t);y^0(x_t),p_{t},x_t\right) -h(x_t,\xi_{t}).$$ 
Thus, for all $t$ and  each $ j=1,\dots ,J$, 
\begin{equation}\label{eq regression}
g_{j}\Bigl( z^{\ast }\left( s^{\ast }(x_t);p_{t},x_t,\xi _{t}\right),y^0(x_t),x_t \Bigr) = f_j(x_t,p_t) - e_{jt}
\end{equation}%
where 
$f_j(x_t,p_t) \equiv \sigma
_{j}^{-1}\left( s^{\ast }(x_t);y^0(x_t),p_{t},x_t\right)$
and
$e_{jt}\equiv h_j(x_t,\xi_{jt}).$
By Lemma \ref{thm g identified} the left side of (\ref{eq regression}) is known (recall that the values of each $z^{\ast }\left( s^{\ast }(x_t);p_{t},x_t,\xi _{t}\right) $ are observable, even though the value of each $\xi _{t}$ is not). Thus, for each $j$ this equation takes the form of a separable nonparametric regression model. Given Assumption \ref{ass completeness}, identification of each function $f_j $ follows immediately from the identification result of \cite*{NeweyPowell2003}.  This implies identification of each $e_{jt}$ (i.e., $h_j(x_t,\xi_{jt})$) as well.
\end{proof2}

Identification of the conditional demand functions $\bar{\mathscr{s}}(\cdot;t)$ for all $t$ now follows easily.

\begin{theorem}\label{thm conditional demand}
Under Assumptions \ref{ass index}--\ref{ass completeness}, $\bar{\mathscr{s}}(\cdot;t)$ is identified on $\mathcal{C}(x_t,\xi_t)$ for all $t$.
\end{theorem}

\begin{proof2}
Recall that
\begin{align*}
  \bar{\mathscr{s}}(Z_{it},Y_{it},P_t;t)&=\mathscr{s}(Z_{it},Y_{it},P_t,x_t,\xi_t)\\
                                        &= \sigma(g(Z_{it},Y_{it},x_t)+h(x_t,\xi_t),Y_{it},P_t,x_t)\\
                                        &= E\left[Q_{it}\vert Z_{it},Y_{it},P_t,x_t,h(x_t,\xi_t)\right].
\end{align*}
Because $Q_{it},Z_{it},Y_{it},P_t,X_t$ are observed and each $h(x_t,\xi_t)$ is known, the result follows.
\end{proof2}

We emphasize that although the conditional demand functions  $\bar{\mathscr{s}}(\cdot;t)$ are indexed by $t$, this merely stands in for the values of $X_t$ and $h(X_t,\Xi_t)$. Within a single market, there is no price variation. However, Lemma \ref{lem residuals identified} allows us to utilize information from all markets with same values of $X_t$ and $h(X_t,\Xi_t)$ to reveal how price variation affects demand at all  $(Z_{it},Y_{it},P_t,X_t,h(X_t,\Xi_t))$ in their joint support.

\subsection{Identification of Demand}\label{sec id demand}

As discussed already, knowledge of the conditional demand functions suffices for a large fraction of the questions motivating demand estimation, but not all. In particular, it is not sufficient to answer questions concerning effects of $X_t$ on demand or other counterfactual outcomes when $X_t$ changes holding $\Xi_t$ fixed. Addressing such questions will require separating the impacts of $X_t$ and $\Xi_t$. This can be done by adding the standard assumption that $X_t$ is exogenous.

\begin{assumption}[Exogenous Product Characteristics]\label{ass exog x}
  $E[\Xi_t|X_t]=0$.
\end{assumption}

When Assumption \ref{ass exog x} holds,  the definition (\ref{eq define h}) implies
$$h(X_t,\Xi_t)=\Xi_t.$$
This has two important implications.  First, when Assumption \ref{ass exog x} is maintained, the IV exclusion condition (part (i) of  Assumption \ref{ass IV}) softens to require instruments $W_t$ that are exogenous conditional on exogenous (rather than endogenous) $X_t$. Second, Lemma \ref{lem residuals identified} now implies that each realization $\xi_t$ of the demand shock vector is identified. Recalling that
\begin{equation*}
  \mathscr{s}(Z_{it},Y_{it},P_t,X_t,\Xi_t)= E\left[Q_{it}\vert Z_{it},Y_{it},P_t,X_t,\Xi_t\right],
\end{equation*}
identification of $\mathscr{s}$ follows immediately from the facts that $(Q_{it},Z_{it},Y_{it},P_t,X_t)$ are observed and all realizations of $\Xi_t$ are now known.

\begin{theorem}\label{thm demand}
Under Assumptions \ref{ass index}--\ref{ass exog x}, $\mathscr{s}$ is identified on $\mathcal{C}$.
\end{theorem}

\section{Discussion}\label{sec discuss}

The results above demonstrate nonparametric identification of demand
(and conditional demand) using a combination of within-market and
cross-market variation.  Compared to a setting with market-level data,
 micro data can (i) permit
demand specifications that condition on consumer-level observables, (ii) avoid  the need to  restrict  how market/product observables $X_t$ enter, and (iii) substantially reduce the reliance on instrumental
variables. 

The last of these may be especially important. The number of instruments needed is halved with micro data,
and there is no  need for the so-called ``BLP instruments.'' This softening of
instrumental variables requirements is achieved because consumer-level
observables create within-market variation in consumers' choice
problems. Such variation is similar in some ways to that which can be
generated by instruments for quantities (see \cite{Berry-Haile-market,Berry-Haile-HBK}).
However, the reason micro-data variation is free from confounding effects of variation in market-level demand shocks is not an assumed exclusion condition in the cross-section of markets but, rather, 
the fact that within a single market these shocks
simply do not vary. Thus, our insights here have a connection to
those underlying ``within'' identification of slope parameters in
panel data models with fixed effects.

The most important message from these results is that identification
of demand for differentiated products follows using the same sorts of
quasi-experimental variation relied upon in simpler settings.  Indeed,
the exploitation of within-unit variation and instrumental variables
are arguably the bread and butter of empirical economics. Of course,
these conclusions lead to several questions about appropriate
instruments, the potential for softening some conditions for
identification by strengthening others, and extensions of our results
to other types of demand models. We discuss these questions in the
remainder of this section. 

\subsection{What Are Appropriate Instruments?}\label{sec IV}

The fact that reliance on instruments is standard does not imply that
instruments will always be available. Rather, this merely shifts
discussion of identification largely to standard questions concerning
the availability of suitable instruments. What are likely instruments
in practice?

Candidate instruments for prices include most of those typically
relied upon in the case of market-level data (see
\cite*{Berry-Haile-annreview} for a more complete discussion of these candidate instruments). Classic
instruments for prices are cost shifters that are excluded from the
demand system and (mean-) independent of the demand shocks
$\Xi_t$. When cost shifters are not observed, proxies for cost
shifters may be available and can satisfy the required exclusion
conditions.\footnote{Examples of such proxies, plausibly exogenous in some
  applications, are so-called ``Hausman instruments,'' i.e., prices of
  the same good in other markets (e.g., \cite*{Hausman-Leonard-Zona},
  \cite*{Hausman96}, or Nevo (2000, 2001)).}\nocite{Nevo2000b}
Exogenous shifters of market structure (e.g., exogenous merger
activity or, in some cases, exogenous variation in common ownership) that affect prices through equilibrium markups can also
serve as instruments. 

Micro data can also result in availability of a
related category of candidate instruments: market-level observables (e.g., market-level demographic measures)
that alter equilibrium markups. Berry and Haile (2014, 2016) refer to
these as ``Waldfogel'' instruments, after
\cite*{Waldfogel_whom}.\footnote{See also \cite*{GentzkowShapiro},
  \cite*{Fan}, and \cite*{LiHartmannAmano}.}  When micro data are
available, we can directly account for the impacts of
individual-specific demographics, so it may be reasonable to assume
that market-level demographics are excluded from the conditional
demands we seek to identify. The requirement that these market-level
measures be mean independent of the market-level demand shocks is a
significant assumption, ruling out certain kinds of geographic sorting
or peer effects, for example. But in many applications such an
assumption may be natural.

The exclusion restriction
that defines an appropriate instrument (part (i) of Assumption
\ref{ass IV}), requires $W_t$ to be (mean) independent
of the structural error conditional on $X_t$. This
conditional independence assumption does not require exogenous
product/market characteristics $X_t$, but neither does it allow 
all models with endogenous $X_t$.  Making use of simple graphical causal
models, Appendix \ref{app DAGs} discusses
a variety of cases in which instruments for
prices remain valid under endogeneity of $X_t$. It also discusses the key
case leading to a failure of the exclusion restriction: when $X_t$ is chosen in response to
both $\Xi_t$ and $W_t$. In such cases, $W_t$ could instead serve as an instrument
for the endogenous components of $X_t$, but we would still need different
instruments for prices. In some cases, such instruments may be obtained through natural timing assumptions---e.g., using only the current-period innovations to input costs as the instruments.

Absent from the discussion of candidate instruments above are the ``BLP instruments''---characteristics of competing products.  These  play an essential role as instruments for quantities when one has only market-level data (\cite{Berry-Haile-market}). The relevant exclusion condition in that case requires not only exogeneity of certain product characteristics, but also restrictions on the way they enter demand. Micro data makes it possible to avoid these requirements, although adding them can allow use of BLP instruments for prices (see section \ref{sec stronger index}). Absent such additional assumptions, however, the BLP instruments are unavailable, even when $X_t$ is assumed exogenous. This can be seen in the key equation (\ref{eq regression}), where $x_t$ appears on the right-hand side for each $j$. Each element of $x_t$ ``instruments for itself'' in these equations, leaving no product characteristics excluded. 

\subsection{What About Stronger Functional Forms?}

In practice, estimation in finite samples is almost always influenced by functional form
assumptions---e.g., the choice of parametric structure, kernel functions, or sieve
basis. Such functional forms enable interpolation, extrapolation,
and bridging of gaps between the variation present in the sample
and that needed for nonparametric point identification. A study of nonparametric identification can reveal whether functional form assumptions play a more essential role in one precise sense. One interpretation of our
results is that only limited nonparametric structure is essential: beyond the nonparametric index structure, our main requirement for identification is adequate  variation through $(Z_t,W_t)$ of dimension equal to the dimension of the endogenous variables (prices and quantities).

But one can also ask how imposing additional structure on the demand
model might allow relaxation of our identification
requirements. Answers to this question may be of direct interest and
can also suggest the sensitivity of identification to particular
conditions. For example, we may feel more comfortable when we know
that identification is robust in the sense that a relaxation of one
condition for identification can be accommodated by strengthening
another.  A full exploration of these potential trade-offs describes
an entire research agenda. But some examples can illustrate three
directions one might go to enlarge the set of potential instruments,
further reduce the number of required instruments, or reduce the
required dimensionality of the micro data. 

For simplicity, our
discussion here will consider the typical case in which $X_t$ is
assumed exogenous, focusing then on identification of demand rather than
conditional demand. Recall that in this case we have
  $h(X_t,\Xi_t)=\Xi_t$. Given our focus on the role of $Z_{it}$, for simplicity we will fix and suppress any additional consumer-level 
  observables $Y_{it}$ in what follows.

\subsubsection{Strengthening the Index Structure}\label{sec stronger index}

Our model avoided any restriction on the way the observables $X_{t}$
enter demand. This contrasts with the structure used by  \cite{Berry-Haile-market} to consider identification with market-level data. There,  for each good $j$, one element of $X_{jt}$ was assumed to enter demand only through the $j$th element of the index vector. In practice, such an assumption is common.  And adding  such a restriction here  can introduce another class of potential
instruments: the exogenous characteristics of competing goods, i.e.,
  ``BLP instruments.''\footnote{The key question is the proper excludability of these instruments, which in general requires more than mean independence between $X_t^{(1)}$ and $\Xi_t$. The ``relevance'' of these measures as instruments for prices reflects the fact that in standard oligopoly
  models each good's markup depends on the characteristics of all substitutes or complements.} 
  
To illustrate this as simply as possible,  partition $X_t$ as 
$(X_t^{(1)},X_t^{(2)})$, where $$X_t^{(1)}=\left(X_{1t}^{(1)},\dots,X_{Jt}^{(1)}\right)\in \mathbb{R}^J.$$

Suppose demand takes the form
\begin{equation}\label{eq demand x in index}
\mathscr{s}\left( Z_{it},P_{t},X_{t},\Xi _{t}\right) =\sigma\left( \gamma\left(Z_{it},X_{t},\Xi _{t}\right) 
,P_{t},X^{(2)}_{t}\right),
\end{equation}
where for $j=1,\dots, J$
\begin{equation}
    \gamma_j\left(Z_{it},X_{t},\Xi _{t}\right)=g_j(Z_{ijt},X_t^{(2)})+\eta_j(X^{(1)}_{jt},X_t^{(2)})+\Xi _{jt}.
\end{equation} 

Compared to our original specification, here we (a) restrict $X_t^{(1)}$ to enter only through the index vector; (b) associate the $j$th components of $Z_{it}$ and $X_{t}^{(1)}$ exclusively with the $j$th element of the index vector; and (c) impose additive separability between $Z_{ijt}$ and $X_{jt}^{(1)}$ within each index.\footnote{Exclusivity of $X^{(1)}_{jt}$ to the index $\gamma_j$ is essential to the point we illustrate here, and this is most natural when exclusivity of each $Z_{ijt}$ differentiates the elements of the index vector. Part (c) substantially simplifies the exposition. As in our more general model, the elements of $\gamma(Z_{it},X_{t},\Xi _{t})$ need not be linked to particular goods.} This specification requires that each element of  $Z_{it}$ can be matched to an element of  $X_t^{(1)}$ that affects demand in a similar way. Many specifications in the literature satisfy this requirement, typically with additional restrictions such as linear substitution between $Z_{ijt}$ and $X_{jt}^{(1)}$. We will also strengthen the common choice probability condition to require existence of a common choice probability vector  $s^*(X_t)$ that does not vary with $X_t^{(1)}$.\footnote{Formally, we assume that for each $x^{(2)}\in \supp X_{t}^{(2)}$, there exists a choice probability vector $s^{\ast }(x^{(2)})$ such that for
all $x^{(1)}\in $ $\supp X_{t}^{(1)}|\{X_{t}^{(2)}=x^{(2)}\}$, 
$s^{\ast }(x^{(2)})\in \mathcal{S}(p,(x^{(1)},x^{(2)}),\xi )$ 
for all $(p,\xi)\in \text{supp}\,\left( P_{t},\Xi _{t}\right) |\{X_{t}=(x^{(1)},x^{(2)})\}$.} For simplicity we also assume that, for each $x^{(2)}$ in the support of $X_t^{(2)}$ there is some point $z^0(x^{(2)})$ common to $\mathcal{Z}(x,x^{(2)})$ for all $x$ in the support of $X_t^{(1)}|\{X_t^{(2)}=x^{(2)}\}$.

For the remainder of this section  we will condition on $X_{t}^{(2)}$ (treating it fully flexibly), suppress it from the notation,  and let $X_{t}$ represent $X_{t}^{(1)}$. 
Posing the identifications question here requires a different set of normalizations.\footnote{Those in section \ref{subsec normal} do not respect the exclusivity restrictions imposed here and, therefore, cannot be assumed without loss in this case.} These are standard location and scale normalizations. First, because adding a constant $\kappa_j$ to $g_j$ and subtracting the same constant from $\eta_j$ would leave the model unchanged, we take an arbitrary $x^0\in\mathcal{X}$ and set \begin{equation}\label{eq norm eta}
    \eta_j(x^0_j)=0 \quad \forall j.
\end{equation} 
Even with (\ref{eq norm eta}) (and our maintained $E[\Xi_t]=0$), it remains true that any linear (or other injective) transformation of the index $\gamma_j$ could offset by an appropriate adjustment to the function $\sigma$, yielding multiple representations of the same demand system (recall the related observation in section \ref{subsec normal}). Thus, without loss, we normalize the location and scale of each index $\gamma_j$ by setting
\begin{eqnarray*}
g_j\left( z_j^{0}\right)  &=&0 \\
\frac{\partial g_j\left( z_j^{0}\right) }{\partial z_j}&=&1.
\end{eqnarray*}

The arguments in Lemmas  \ref{lem open ball}--\ref{thm g identified} now demonstrate identification of each function $g_j$. At the common choice probability vector $s^*$, the inverted demand system takes the form of equations
\[
g_{j}\left( z_{j}^{\ast }\left( s^{\ast }\right) \right) +\eta _{j}\left(
x_{jt}\right) +\xi _{jt}=\sigma _{j}^{-1}\left( s^{\ast };p_{t}\right) 
\]%
for each $j$. Writing the $j$th equation as
\begin{equation}\label{eq BLP IV}
g_{j}\left( z_{j}^{\ast }\left( s^{\ast }\right) \right) =-\eta _{j}\left(
x_{jt}\right) +\sigma _{j}^{-1}\left( s^{\ast };p_{t}\right) -\xi _{jt},
\end{equation}
we obtain a nonparametric regression equation with RHS variables $x_{jt}$ and $p_t$. Here $x_{-jt}$ is excluded, offering $J-1$ potential instruments for the endogenous prices $p_t$.  Thus,  one additional instrument---e.g., a scalar market-level cost shifter or Waldfogel instrument---would yield enough instruments to  obtain identification of the unknown RHS functions and  the ``residuals'' $\xi _{jt}$.\footnote{Here  the  separability in $X_{jt}$  provides a falsifiable  restriction.}
As before, once these demand shocks are identified, identification of demand follows immediately.%

Many variations on this structure are possible.  For example, as in many empirical specifications, one might assume that  $p_{jt}$ enters demand only through the $j^{th}$ index. This can lead to a regression equation (the analog of (\ref{eq BLP IV})) of the form 
\[
g_{j}\left( z_{j}^{\ast }\left( s^{\ast }\right) \right) =-\eta _{j}\left(
x_{jt},p_{jt}\right) +\sigma _{j}^{-1}\left( s^{\ast}\right) -\xi _{jt}.
\]
Now only one instrument for price is necessary. For example, the BLP instruments can overidentify demand.

\subsubsection{A Nonparametric Special Regressor}\label{sec special reg}

A different approach is to assume that the
demand system of interest is generated by a random utility model with
conditional indirect utilities of the form
\[
U_{ijt}=g_{j}(Z_{ijt})+\Xi _{jt}+\mathcal{E}_{ijt},
\]%
where $\mathcal{E}_{ijt}$ is a scalar random variable whose nonparametric distribution
depends on $X_{jt}$ and $P_{jt}$ (equation (\ref{eq def mu}) gives a parametric example). In this case, our Lemma \ref{thm g identified} demonstrates
identification of each function $g_{j}(\cdot)$ up to a normalization of utilities.

If one is willing to add the assumption of independence between $Z_{ijt}$ and $\mathcal{E}_{ijt}$, this turns $g_{j}(Z_{ijt})$ into
a known special regressor. Under a further  (typically very restrictive) large support
assumption on $g_{j}(Z_{j})$, a standard argument demonstrates identification of the marginal
distribution of $(\Xi_{jt}+\mathcal{E}_{ijt})|(X_t,P_t)$. This is not sufficient to identify demand. However,
one can use these
marginal distributions to define a nonparametric IV regression equation for
each choice $j$, where the LHS is a conditional mean and $\Xi_{jt}$ appears on the RHS as an additive structural error.\footnote{See our earlier working paper, \cite*{BerryHaileNPold}.}  In each of these equations the prices and characteristics of
goods $k\neq j$ are excluded. Identification of these equations identifies all demand shocks, and identification of demand then follows as in Theorem \ref{thm demand}.  Thus, in this framework one needs only one
instrument for price, and exogenous characteristics of competing
goods (BLP IVs) would be available as instruments.

\subsubsection{A Semiparametric Model}\label{sec semi nested}

Moving further in the direction of parametric models commonly used in practice can  reduce both the required dimensionality of consumer attributes and the number of required instruments. 
As one example, consider a semi-parametric nested logit model. We condition on (and suppress from the notation) $X_t$,%
\footnote{%
By conditioning on  $X_{t}$, we permit it to enter the model
fully flexibly. Let conditional indirect utilities take the form
\[
u_{ijt}=u\left( x_{t},g_{j}\left( z_{it},x_{t}\right) +\xi _{jt}-\alpha
\left( x_{t}\right) p_{jt}+\mu _{ijt}\left( x_{t}\right) \right) ,
\]%
where $u$ is strictly increasing in its second argument, $\alpha \left(
x_{t}\right) $ is arbitrary, and $\mu _{ijt}\left( x_{t}\right) $ is a
stochastic component taking the standard composite nested-logit form at each
$x_{t}$. The identification argument sketched here may be repeated at each $x_t$.} and consider a semiparametric nested logit model where inverse demand  in  market $t$, given $z_{it}$, is
\begin{equation}
g_{j}(z_{it})+\xi _{jt}=\ln (s_{jt}(z_{it})/s_{0t}(z_{it}))-\theta \ln
(s_{j/n,t}(z_{it}))+\alpha p_{jt}.  \label{eq nl inverse}
\end{equation}%
Here $s_{jt}(z_{it})$ denotes good $j$'s observed choice probability in market $t$ conditional on $z_{it}$, and $s_{j/n,t}(z_{it})$ denotes its within-nest conditional choice probability. The scalar $\theta $ denotes
the usual \textquotedblleft nesting parameter.\textquotedblright\ Here we allow $z_{it}$ to have fewer than $J$ elements.

As with the standard representation of most parametric models of inverse demand, the  nested logit model embeds   normalizations of the indices and demand function analogous to our choices of $A(x)$ and $B(x)$ in section \ref{subsec normal}. However, we must still normalize the location of either $\Xi_{jt}$ or $g_j$ for each $j$ to pose the identification question. Here we will set $g_j(z^0) = 0$ for all $j$, breaking with our prior convention by leaving each $E[\Xi_{jt}]$ free.

 Take any market $t$ and any $z\in $ $\mathcal{Z}$. Differentiating
\eqref{eq
nl inverse} with respect to one (possibly, the only) element of $z_{it}$%
---say $z^{(1)}_{it}$---at the point $z$ yields
\begin{equation}
\frac{\partial g_{j}(z)}{\partial z^{(1)}}=\frac{\partial \ln s_{jt}(z)}{%
\partial z^{(1)}}-\frac{\partial \ln s_{0t}(z)}{\partial z^{(1)}}-\theta \frac{%
\partial \ln s_{j/nt}(z)}{\partial z^{(1)}}.  \label{eq nl diff}
\end{equation}%
In this equation,$\frac{\partial g_{j}(z)}{\partial z^{(1)}}$ and $\theta $
are the only unknowns. Moving to another market $t^{\prime }$, we can obtain
a second equation of the same form in which the left-hand side is identical to that in (%
\ref{eq nl diff}). Equating the right-hand sides yields
\[
\frac{\partial \ln s_{jt}(z)}{\partial z^{(1)}}-\frac{\partial \ln s_{0t}(z)}{%
\partial z^{(1)}}-\theta \frac{\partial \ln s_{j/nt}(z)}{\partial z^{(1)}}=\frac{%
\partial \ln s_{jt^{\prime }}(z)}{\partial z^{(1)}}-\frac{\partial \ln
s_{0t^{\prime }}(z)}{\partial z^{(1)}}-\theta \frac{\partial \ln
s_{j/nt^{\prime }}(z)}{\partial z^{(1)}}.
\]%
Thus, we can solve for $\theta $ as long as%
\[
\frac{\partial \ln s_{j/nt}(z)}{\partial z^{(1)}}\neq \frac{\partial \ln
s_{j/nt^{\prime }}(z)}{\partial z^{(1)}},
\]%
a condition that will typically hold when $\xi _{t^{\prime }}\neq \xi _{t}$
or $p_{t^{\prime }}\neq p_{t}$, and which is directly observed. With $\theta
$ known, we then identify (indeed, over-identify) all derivatives of $%
g_{j}(z)$ from (\ref{eq nl diff}), yielding identification of the function $g$ as in Lemma \ref{thm g identified}. Identification of the remaining parameter
$\alpha $ can then be obtained from the ``regression'' equation 
\begin{equation}
g_{j}(z_{it})=\ln (s_{jt}(z_{it})/s_{0t}(z_{it}))-\theta \ln
(s_{j/n,t}(z_{it}))+\alpha p_{jt}-\xi _{jt}, 
\end{equation}%
obtained from (\ref{eq nl inverse}), using a single
excluded instrument---e.g., an excluded exogenous market-level cost shifter or markup shifter that affects all prices. The constant recovered in this regression represents $E[\Xi_{jt}]$. 

Although this example involves a model that is more flexible than
nested logit models typically estimated in practice, it moves a considerable
distance from our fully nonparametric model. But this example makes clear that additional structure can further reduce the dimension of the required exogenous variation. Indeed, here we can obtain
identification with a single instrument and a
scalar individual-level observable $z_{it}$. This compares to the usual requirement of
two instruments for the fully parametric nested logit when one has only market-level data (see \cite*{Berry94}).  Other semiparametric models may offer more intermediate points in the set of feasible trade-offs between the flexibility of the model and the dimension of exogenous variation needed for identification.

\subsection{What about Continuous Demand Systems?}

Although we have focused on the case in which the consumer-level
quantities $Q_{ijt}$ are those arising from a discrete choice model,
nothing in our proofs requires this. In other settings, the demand
function $\mathscr{s}$ defined in (\ref{eq expected demand}) may simply be
reinterpreted as the expected vector of quantities demanded
conditional on $(X_t,P_t,\Xi_t,Z_{it},Y_{it})$.%
\footnote{Note that the demand faced by firms in market $t$ is the
  expectation of this expected demand over the joint distribution of
  $(Z_{it},Y_{it})$ in the market.} %
Applying our results to continuous demand is therefore
just a matter of verifying the suitability of our assumptions.\footnote{\cite*{BerryGandhiHaile} describe a broad class of continuous choice models
that can satisfy the key injectivity property of Assumption \ref{ass
invertible demand}. These can include mixed continuous/discrete
settings, where individual consumers may purchase zero or any positive
quantity of each good.}

As one possibility, consider a \textquotedblleft mixed CES\textquotedblright\
model of continuous choice, similar to the model in \cite*{AdaoCostDonaldson17}, with $J+1$ products. Here we reintroduce $Y_{it}$ to denote consumer $i$'s income, measured in units of the numeraire good $0$. Each consumer $i$ in market $t$ has utility
over consumption vectors $q\in \mathbb{R}_{+}^{J+1}$ given by
\[
u\left( q;z_{it},x_t,p_t,\xi_t \right) =\left( \sum_{j=0}^{J}\phi_{ijt}q_{j}^{\rho}\right) ^{1/\rho },
\]%
where $\rho \in \left( 0,1\right) $ is a parameter and each $\phi _{ijt}$
represents idiosyncratic preferences of consumer $i$. Normalizing $\phi _{i0t}=1$, let
\[
\phi _{ijt}=\exp \left[ \left( 1-\rho \right) \left( g_{j}\left(
z_{it},x_t\right) +\xi_{jt}+x_{jt}\beta_{it}\right) \right] \text{, }j=1,\dots ,J,
\]%
where $\beta _{it}$ is a random vector with distribution $F$ representing consumer-level preferences for product
characteristics. With $p_{0t}=1$, familiar
CES algebra shows that Marshallian demands are
\begin{equation}\label{eq CES demand}
q_{ijt}=\frac{y_{it}\exp \left( g_{j}\left( z_{it},x_t\right) +\xi
_{jt}+x_{jt}\beta _{it}-\alpha \ln (p_{jt})\right) }{1+\left[
\sum_{k=1}^{J}\exp \left( g_{k}\left( z_{it},x_t\right) +\xi _{kt}+x_{kt}\beta_{it}-\alpha \rho \ln (p_{kt})\right) \right] },
\end{equation}
where $\alpha=\slashfrac{1}{(1-\rho)}$. Equation (\ref{eq CES demand}) resembles a choice probability for a random coefficients logit model, although
the quantities $q_{it}$ here take on continuous values and do not sum to one. It is easy to show that our Assumptions \ref{ass index}--\ref{ass
linear index} are satisfied for the expected CES demand functions, which
take the form
\[
\sigma _{t}(g(z_{it},x_t)+\xi _{t},y_{it},x_t, p_{t})=E\left[Q_{it}\vert z_{it},y_{it},p_{t},x_t,\xi _{t}\right],
\]%
where the $j$th component of $E[Q_{it}|z_{it},y_{it},x_t,p_{t},\xi _{t}]$ is
\begin{equation*}
\int \frac{y_{it}\exp \left( g_{j}\left( z_{it},x_t\right) +\xi _{jt}+x_{jt}\beta_{it}-\alpha \ln (p_{jt})\right) }{1+\left[ \sum_{k=1}^{J}\exp
\left( g_{k}\left( z_{it},x_t\right) +\xi _{kt}+x_{kt}\beta_{it}-\alpha \rho \ln (p_{kt})\right) \right] }\,dF(\beta_{it}).
\end{equation*}

\section{Lessons for Applied Work}\label{sec lessons}

Although the study of identification is formally a theoretical exercise, a primary motivation for our analysis is to provide guidance for the practice and evaluation of demand estimation in applied work.  Here we discuss some key messages.

\subsection{The Incremental Value of Micro Data}

The most important practical lesson from our results is that the marginal value of micro data is high. It is not surprising that a setting allowing one to exploit variation both across markets and within markets is more informative than one with only cross-market variation. But the specific benefits of micro data concern some of the most significant challenges to identification of demand when one has only market-level data: (i) the need to instrument for all prices and quantities, and (ii) the nonparametric functional form and exogeneity conditions that allow these IV requirements to be satisfied. We have shown that adding micro data can eliminate the need to instrument for quantities and, therefore, the necessary reliance on BLP instruments. This, in turn, avoids the need for any restriction on the way observables at the level of the product and market enter the model. Furthermore, our results on identification of conditional demand imply that one can often obtain price elasticities without any exogenous product characteristics, much less the use of such characteristics as instruments. 

These are significant advantages. Researchers should, therefore, not only prefer micro data, but should seek it out whenever possible. Of course, even when the setting and assumptions permit use of BLP instruments---or when the micro data available are more limited than we have assumed to explore fully nonparametric identification---variation from micro data can be powerful. This message is consistent, for example, with the findings in the empirical literature (e.g., \cite{Petrin02}) that the addition of even limited micro data often results in much more precise estimates than those obtained with market-level data alone.

\subsection{The Necessity of Cross-Market Variation}

Another important lesson from our work concerns the need for cross-market variation,
even when one has micro data. Variation within a
single market cannot suffice for identification, at least without additional
assumptions.

Formally, our proofs relied on cross-market variation, even for identification of the function $g(\cdot,y^0(x),x)$ (see section \ref{sec id index}).  But the necessity of cross-market variation is also easy to see. In a single market the observables consist of conditional choice
 probability vectors  $s\left( z_{i},y_{i}\right)$ at all
$\left(z_{i},y_{i}\right)\in \Omega$---here we will suppress the index $t$ as well as the
observables $(X_{t},P_{t})$, since these have no variation in a single market.
Consider an arbitrary (and, thus, typically mis-specified) invertible parametric demand function $\sigma\left(g(z_{i},y_{i})+\xi ;\theta \right)$ that maps $J$ indices $g_j(z_{i},y_{i})+\xi_j$ to market shares. For concreteness, suppose this is a nested logit model with  ``mean utilities'' $g_j(z_{i},y_{i})+\xi_j$ and nesting parameter(s) $\theta$. By standard results (see \cite{Berry94}), given any
value of $\theta $, this model can fit the data in the
market perfectly by setting $\xi_j =0$ and
$$
g\left( z_{i},y_{i}\right) =\sigma ^{-1}\left( s\left( z_{i},y_{i}\right);\theta \right).
$$
This yields a different function $g$ for every candidate value of $\theta$, and no value of $\theta $
can be ruled out. Thus, the observables from a single market cannot identify the nesting parameters in this semiparametric nested logit model, much less determine whether the nested logit structure is correct.

It is also easy to see here how having micro data in multiple markets can help.
When the same demand model is assumed to apply to multiple markets,
the same $(g,\theta)$ pair must fit the data in each market. The
resulting restrictions can rule out incorrect candidates for $\theta$
and $g$, as we have seen in section \ref{sec semi nested}.\footnote{As
  suggested in section \ref{sec semi nested},  with a scalar $\theta$
  defining substitution between products in response to changes in the
  index vector, two markets may suffice. Our Lemmas 1--3 show how the
  restrictions across many markets allow identification when these
  substitution patterns are nonparametric.}

Of course, although we have suppressed the price vector $P_t$ when talking about single market, the effects of price variation on quantities demanded are essential. Because price vectors are typically fixed within markets by definition, exogenous sources of cross-market price variation will be needed.
Thus, even in the presence of micro data there are at least two
reasons applied researchers
should seek out data on multiple markets. First, a combination of within-market and cross-market
variation is needed to identify flexibly-specified effects of consumer
observables (including the function $g$). Second, cross-market variation through instruments for prices is essential for learning how demand responds to price variation.

These observations also serve as a caution. As a practical matter, with a fully parametric specification of demand it will often be possible to estimate all parameters with data from only one market.  And in some cases, only a single market is available for study. The classic work of \cite{McFadden-Trans} offers one example. However, identification
in such cases will implicitly rely on functional form restrictions---restrictions that could be
relaxed in a multi-market setting.

Our findings on the theoretical importance of cross-market variation can be linked to the practical findings of \cite*{MicroBLP}, who reported that when using only consumer-level variation---no cross-market variation or ``second choice'' data\footnote{ \cite*{Berry-Haile-HBK} discuss the close relationship between second-choice data and micro data from two markets.}---their attempts to
  estimate random coefficients logit models
  failed due to a nearly flat objective function. They speculated (p. 90) that
  ``in applications to other data sets, variation in the choice set
  (either over time or across markets) might provide the information
  necessary to estimate the random coefficients.'' Our results provide a  nonparametric confirmation of
  that conjecture, again pointing to the practical value of data that combine within- and cross-market variation.

\subsection{What Does Not Follow}

Although nonparametric identification results can offer important
insights, they address a very specific question about what can be
learned from data. A nonparametric identification result can
demonstrate a particular sense in which parametric assumptions are not
essential. But this does not mean that parametric (or other)
assumptions relied on in practice can be ignored. The choice of
finite-sample approximation method can of course matter. In the case
of demand estimation, functional form restrictions used in practice
restrict the families of demand functions considered in a way that can
constrain the answers to key questions. Thus, sensitivity of estimates
(most importantly, estimates of the quantitative answers to the
economic questions of ultimate interest) to functional form choices
remains an important issue for empirical researchers to explore. Likewise, it remains important to explore new (parametric, semiparametric, or nonparametric)
estimation approaches. Our nonparametric identification results ensure
that such explorations are possible and may even suggest new estimation strategies.

We also emphasize that our sufficient conditions for nonparametric
identification should not be viewed as necessary conditions, formally
or informally, for demand estimation in practice. Nonparametric identification
results should guide our thinking about the strength of the available
data and empirical results. But it would be a mistake to view these as
conditions that must be confirmed before proceeding with empirical work. Nonparametric identification of most
models in economics (even regression models) relies on
assumptions---index assumptions, separability assumptions,
completeness conditions, support conditions, monotonicity conditions,
or other shape restrictions---that will often (perhaps typically) fall
short of full satisfaction in practice. Conditions for nonparametric
identification are not a hurdle but an ideal---a point of reference
that can guide our quest for and aid our assessment of the best
available empirical evidence.

\section{Conclusion}\label{sec conclusion}

Since \cite*{BLP}, there has been an explosion of interest in
estimation of demand models that incorporate both flexible
substitution patterns and explicit treatment of the demand shocks that
introduce endogeneity/simultaneity. Understandably,
this development has been accompanied by questions about
what allows identification of these models. Our results,  here and in \cite{Berry-Haile-market}, offer a reassurance that identification follows
from traditional sources of quasi-experimental variation in the form
of instrumental variables and panel-style within-market
variation. This reassurance is particularly important because of the
wide relevance of these models to economic questions and the depth of the identification
challenge in the context of demand systems---notably, the fact that
even purely exogenous variation in prices is generally not sufficient to
identify price elasticities or other essential features of demand.\footnote{See the discussion in
  \cite*{Berry-Haile-HBK}.}

Furthermore, identification of these models is not
fragile. Identification does not rely on
``identification-at-infinity''  arguments; it is not limited to
particular types of settings (e.g., random utility discrete choice);
one can substitute one type of variation for another (e.g., replacing instruments for quantities with micro-data variation), depending on the
type of data available; and one can relax some key
conditions by strengthening others. Thus, although this is a case in
which identification results come well after an extensive empirical
literature has already developed, the nonparametric foundation for
this literature is strong.

\appendix

% Tikz settings optimized for causal graphs.
\usetikzlibrary{shapes,decorations,arrows,calc,arrows.meta,fit,positioning}
\tikzset{
    -Latex,auto,node distance =1 cm and 1 cm,semithick,
    state/.style ={ellipse, draw, minimum width = 0.7 cm},
    point/.style = {circle, draw, inner sep=0.04cm,fill,node contents={}},
    bidirected/.style={Latex-Latex,dashed},
    el/.style = {inner sep=2pt, align=left, sloped}
}

\numberwithin{equation}{section}
\appendixpage

\section{Excludability Conditional on \\ Endogenous Product Characteristics}\label{app DAGs}

In section \ref{sec IV} we discussed several categories of instruments
$W_t$ commonly relied upon to provide exogenous variation in prices.
Here we discuss the question of when such instruments remain properly
excluded even when conditioning on observables $X_t$ that are not
independent of $\Xi_t$.  Such instruments are required for Theorem
\ref{thm conditional demand} to apply when Theorem \ref{thm demand}
does not, allowing identification of conditional demand without
requiring exogeneity of $X_t$ (or instruments for $X_t$).

In what follows we suppress the market subscripts $t$ on the random variables
$X_t, W_t,  \Xi_t$, etc.  
Our discussion will utilize graphical causal models, with d\nobreakdash-separation providing the key criterion for
assessing the independence between $W$ and $\Xi$ conditional on $X$.%
\footnote{See, e.g., \cite*{Pearl-Causality} and \cite*{Pearl_etal_Primer}, including references therein. Throughout we maintain the standard assumption that nodes in the causal directed acylic graphs are independent of their nondescendants conditional on their parents.} %
Our use of these tools is elementary. However, the graphical approach allows transparent treatment of many possible economic examples
inducing a smaller number of canonical dependence structures. It also can be
highly clarifying when one ventures beyond the simplest
cases. Following the literature on graphical causal models, we focus on full
conditional independence,
\begin{equation}\label{eq cond indep}
W \Ind \Xi \,\vert X,
\end{equation}
which of course implies the conditional mean independence required by Theorem \ref{thm conditional demand}.

Below we first discuss several causal graphs (and motivating economic examples) that ``work''---i.e., that imply the conditional independence condition (\ref{eq cond indep}). We then discuss the main type of structure that does not work---i.e., in which  (\ref{eq cond indep}) fails despite unconditional independence between $W$ and $\Xi$.
We will see that each type of instrument discussed in section \ref{sec IV}  can remain valid under several models of endogenous $X$.  However,  each type of instrument can also fail; in particular,  (\ref{eq cond indep}) will fail despite unconditional independence between $W$ and $\Xi$ when firms choose $X$ in ways that depend on both $W$ and $\Xi$ (or their ancestors). However, in these situations,  a natural timing assumption can often yield a new set of valid instruments for prices.

\subsection{Graphs that Work}

\newcommand{\pa}{\textrm{pa}}

\subsubsection{Fully Exogenous Instruments}\label{sec fully exog}

The simplest cases arise when the instruments $W$ satisfy
\begin{equation}\label{eq fully exog}
  W \Ind  (X,\Xi).
\end{equation}
The conditional independence condition (\ref{eq cond indep}) is then immediate, regardless of any dependence between $X$ and $\Xi$.

For example, suppose $X$ is chosen by firms with knowledge of $\Xi$,
so that $X$ is endogenous in the same sense that prices are.  Given
(\ref{eq fully exog}), one obtains the causal graph shown
in Figure \ref{fig exoga}.%
\footnote{\label{fn ancestral}We assume throughout that prices
  and quantities are not among the ancestors of $(X,W,\Xi)$ and,
  therefore, typically exclude them from the graphs without loss when examining
  the properties of the joint distribution of
  $(X,W,\Xi)$. %
  Exclusion of prices and quantities from the ancestors
  of $(X,W,\Xi)$ is implied by standard assumptions that consumers take
  $X$ and $\Xi$ as given when making purchase decisions,  that $W$
  does not respond to prices or quantities, and that prices are not chosen before $X$. Note that in the case of simultaneously chosen $X$ and $P$, the parents of $X$ and $P$ will include only previously determined variables (namely, those entering the reduced forms for $X$ and $P$); thus, neither $X$ not $P$ will be an ancestor of the other. We include an example of simultaneous determination below.} % end of long footnote
The conditional
independence condition (\ref{eq cond indep}) then can also be seen to
follow immediately by the d-separation criterion. We of course reach the same conclusion if the direction of
causation between $\Xi$ and $X$ is reversed, as in Figure \ref{fig
  exogb}---e.g., if $X$ is chosen without knowledge of $\Xi$ but the
distribution of $\Xi$ changes with the choice of $X$. Taking the
canonical example of demand for automobiles, a manufacturer's choice
to offer a fuel efficient hybrid sedan  may imply a very
different set of (or response to) relevant unobserved characteristics
than had a pickup truck  or luxury SUV been offered
instead.

\begin{figure}[h!]\caption{}
\label{fig exoga}
\centering
\begin{tikzpicture}
    % node set with absolute coordinates
    \node[state] (xi) at (0,0) {$\Xi$};

    % Locations can be:
    % right,left,above,below,
    % above left,below right, etc
    \node[state] (x) [right =of xi] {$X$};
    \node[state] (w) [right =of x] {$W$};

    % Directed edge
    \path (xi) edge (x);

\end{tikzpicture}
\end{figure}

\begin{figure}[h!]\caption{}
\label{fig exogb}
\centering
\begin{tikzpicture}
    % node set with absolute coordinates
    \node[state] (xi) at (0,0) {$\Xi$};

    % Locations can be:
    % right,left,above,below,
    % above left,below right, etc
    \node[state] (x) [right =of xi] {$X$};
    \node[state] (w) [right =of x] {$W$};

    % Directed edge
    \path (x) edge (xi);

\end{tikzpicture}
\end{figure}

A similar structure is obtained when dependence between $\Xi$ and $X$ reflects a common cause (which could be latent). This is illustrated in Figure \ref{fig lag fork}.  For example, the common cause $V$  might represent past values of demand shocks, which are predictive of the current shocks $\Xi$ and  are  determinants of firms' choices of product characteristics $X$.

\begin{figure}[h!]\caption{}\label{fig lag fork}
\centering
\begin{tikzpicture}
    % node set with absolute coordinates
    \node[state] (xi) at (0,0) {$\Xi$};

    % Locations can be:
    % right,left,above,below,
    % above left,below right, etc
    \node[state] (x) [right =of xi] {$X$};
    \node[state] (w) [right =of x] {$W$};
    \node [state] (v) [above = of xi] {$V$};

    % Directed edge
    \path (v) edge (xi);
    \path (v) edge (x);

\end{tikzpicture}
\end{figure}

All the instrument types discussed in section \ref{sec IV} can  satisfy  (\ref{eq fully exog}).  For example, $W$ could represent cost shifters such as input price shocks, realized after $X$ is chosen and not affected by $X$. These might be shocks to import tariffs; shipping costs; retailer costs (e.g., rents, wages); or prices of manufacturing inputs.

One can also obtain this structure when $W$ represents exogenous shifters of markups. Mergers (full or partial) that are independent of $\Xi$ and leave product offerings unchanged offer one possibility.  Another is cross-market variation in the distribution $F_{YZ}(\cdot|t)$ (or other aggregate demographic measure at the market or regional level), as long as this variation is independent of $\Xi$ (as required generally for the validity of Waldfogel instruments) and $X$.

Proxies for cost shifters can also satisfy (\ref{eq fully exog}). Proxies lead to slightly different causal graphs, as illustrated in Figure \ref{fig proxy ok}.   There $L$ represents latent cost shifters known by firms when setting prices but not affecting $X$. The instruments $W$ are proxies for $L$.  For example, $W$ could represent Hausman instruments when we have both (a) the latent product-level cost shifters $L$ are independent of $X$, and (b) for each market, $\Xi$ and $X$ are independent of prices in other markets.  Condition (b) may often fail, since prices in other markets depend on the product characteristics in those other markets, and the characteristics of a given product will typically be highly correlated across markets. Fortunately, independence between $X$ and $W$ is not required, as we discuss next.

\begin{figure}[h!]\caption{}\label{fig proxy ok}
\centering
\begin{tikzpicture}
    % node set with absolute coordinates
    \node[state] (xi) at (0,0) {$\Xi$};

    % Locations can be:
    % right,left,above,below,
    % above left,below right, etc
    \node[state] (x) [right =of xi] {$X$};
    \node[state] (l) [right =of x] {$L$};
    \node[state] (w) [below =of l] {$W$};

    % Directed edge
    \path (xi) edge (x);
    \path  (l) edge (w);

\end{tikzpicture}
\end{figure}

\FloatBarrier
\subsubsection{Instruments Caused by  X}

Independence between $X$ and $W$ is not required. For example,  consider the case in which (i) $X$ is chosen with knowledge of $\Xi$ and (ii) the choice of $X$ affects $W$. Then we can obtain the following  causal graph, in which (\ref{eq cond indep}) again holds by the d-separation criterion.

\begin{figure}[h!]\caption{}\label{fig x causes w}
\centering
\begin{tikzpicture}
    % node set with absolute coordinates
    \node[state] (xi) at (0,0) {$\Xi$};

    % Locations can be:
    % right,left,above,below,
    % above left,below right, etc
    \node[state] (x) [right =of xi] {$X$};
    \node[state] (w) [right =of x] {$W$};

    % Directed edge
    \path (xi) edge (x);
    \path (x) edge (w);

\end{tikzpicture}
\end{figure}

The case represented by Figure \ref{fig x causes w} allows additional examples of cost shifters beyond those discussed above. For example, suppose $X$ represents product characteristics affecting the level of  labor skill (or quality of another input)  required in production, while $W$ is the producer's average wage.  Alternatively, if producers have market power in input markets, input prices $W$ would be affected by  firms' choices of product characteristics $X$.
\begin{figure}[h!]\caption{}\label{fig proxy still ok}
\centering
\begin{tikzpicture}
    % node set with absolute coordinates
    \node[state] (xi) at (0,0) {$\Xi$};

    % Locations can be:
    % right,left,above,below,
    % above left,below right, etc
    \node[state] (x) [right =of xi] {$X$};
    \node[state] (l) [right =of x] {$L$};
    \node[state] (w) [below =of l] {$W$};
    \node[state] (xio) [right =of l]  {$\Xi_{-t}$};
    \node[state] (xo) [right =of xio] {$X_{-t}$};

    % Directed edge
    \path (xi) edge (x);
    \path (x) edge (l);
    \path  (l) edge (w);
    \path (xio) edge (w);
    \path (xo) edge (w);
    \path[bidirected] (x) edge[bend left=30] (xo);

\end{tikzpicture}
\end{figure}

Allowing dependence between $X$ and $W$  also broadens the applicability of Hausman instruments (or other proxies for latent cost shocks) by allowing the latent costs shocks to depend on $X$. This is illustrated in Figure \ref{fig proxy still ok}, where we have been explicit about the observed characteristics $X_{-t}$ and demand shocks $\Xi_{-t}$ of ``other markets,'' both of which affect the Hausman instruments $W$. Here we link $X$ and $X_{-t}$ with a dotted bidirectional edge to indicate (as usual) dependence through unmodeled common causes. Note that the absence of an edge linking $\Xi$ and $\Xi_{-t}$ reflects an essential assumption justifying Hausman instruments in general (i.e., even when $X$ is exogenous), as does the absence of an edge directly linking $L$ and $\Xi$. Note that here $W$ and $\Xi$ are not independent, but the conditional independence condition (\ref{eq cond indep}) is satisfied, implying the exclusion condition needed for identification.

Reversing the direction of causation between $X$ and $\Xi$ in Figure \ref{fig x causes w} or Figure \ref{fig proxy still ok} leads to the same conclusion. For example, in the first case we obtain structure in Figure \ref{fig fork}, where the required conditional independence condition is again immediate.   Examples generating this structure are similar to those just discussed (including the proxy variation), but with $\Xi$ now representing market-level shocks whose distribution responds to firms' choices of $X$.

\begin{figure}[h!]\caption{}\label{fig fork}
\centering
\begin{tikzpicture}
    % node set with absolute coordinates
    \node[state] (xi) at (0,0) {$\Xi$};

    % Locations can be:
    % right,left,above,below,
    % above left,below right, etc
    \node[state] (x) [right =of xi] {$X$};
    \node[state] (w) [right =of x] {$W$};

    % Directed edge
    \path (x) edge (xi);
    \path (x) edge (w);

\end{tikzpicture}
\end{figure}

\FloatBarrier
\subsubsection{X Caused by Instruments}

In some cases, the conditional independence condition (\ref{eq cond indep}) can hold even when $X$ is affected by $W$. Consider the causal graph in Figure \ref{fig x caused by w}. As an example motivating this structure, suppose $W$ is a product-level cost of producing a product feature measured by $X$, the latter chosen with knowledge of $W$ but before $\Xi$ (or any signal of its realization) is known.

\begin{figure}[h!]\caption{}\label{fig x caused by w}
\centering
\begin{tikzpicture}
    % node set with absolute coordinates
    \node[state] (xi) at (0,0) {$\Xi$};

    % Locations can be:
    % right,left,above,below,
    % above left,below right, etc
    \node[state] (x) [right =of xi] {$X$};
    \node[state] (w) [right =of x] {$W$};

    % Directed edge
    \path (x) edge (xi);
    \path (w) edge (x);

\end{tikzpicture}
\end{figure}

If the product-level cost shifters are latent, proxies---e.g., the same firm's choice of price or product characteristics in other markets---could play the role of $W$, again assuming that $X$ is chosen with knowledge of the cost shock but before $\Xi$ is realized. As usual with such proxies (i.e., even when one assumes $X$ to be exogenous), one must maintain an assumption that the demand shocks $\Xi$ are independent across markets.  Such examples are  illustrated by Figure \ref{x caused by l}.

\begin{figure}[h!]\caption{}\label{x caused by l}
\centering
\begin{tikzpicture}
    % node set with absolute coordinates
    \node[state] (xi) at (0,0) {$\Xi$};

    % Locations can be:
    % right,left,above,below,
    % above left,below right, etc
    \node[state] (x) [right =of xi] {$X$};
    \node[state] (l) [right =of x] {$L$};
    \node[state] (w) [below =of l] {$W$};
    \node[state] (xio) [right =of l]  {$\Xi_{-t}$};
    \node[state] (xo) [right =of xio] {$X_{-t}$};

    % Directed edge
    \path (x) edge (xi);
    \path (l) edge (x);
    \path  (l) edge (w);
    \path (xio) edge (w);
    \path (xo) edge (w);
    \path[bidirected] (x) edge[bend left=30] (xo);

\end{tikzpicture}
\end{figure}

\FloatBarrier
\subsection{Graphs that Don't Work: X is a Collider}

The conditional independence condition (\ref{eq cond indep}) fails when both $W$ and $\Xi$ affect $X$. This is illustrated in Figure \ref{fig collider}. In this case,  $X$ is a collider in the (undirected) path between $W$ and $\Xi$. Thus, although $\Xi$ and $W$ are independent,  (\ref{eq cond indep}) fails.

\begin{figure}[h!]\caption{}
\centering\label{fig collider}
\begin{tikzpicture}
    % node set with absolute coordinates
    \node[state] (xi) at (0,0) {$\Xi$};

    % Locations can be:
    % right,left,above,below,
    % above left,below right, etc
    \node[state] (x) [right =of xi] {$X$};
    \node[state] (w) [right =of x] {$W$};

    % Directed edge
    \path (xi) edge (x);
    \path (w) edge (x);

\end{tikzpicture}
\end{figure}

This structure arises whenever firms' choices of $X$ depend on both $W$ and $\Xi$. An example is when $W$ is a cost shifter affecting firms' choices of $X$, the latter also chosen with knowledge of $\Xi$. Another example is when $W$ is a market-level demographic measure or market structure measure (e.g., product ownership matrix) that, along with $\Xi$, influences firms' choices of  $X$.  As this discussion suggests, this structure arises in most cases in which $X$ and $P$ are chosen simultaneously, with knowledge of both $\Xi$ and $W$. Indeed,  Figure \ref{fig simult collider}, which now includes $P$, provides the  causal graph under such simultaneity. With this structure, Figure \ref{fig collider} is indeed the ancestral graph for $\Xi,X,W$ (see also footnote \ref{fn ancestral}).

\begin{figure}[h!]\caption{}\label{fig simult collider}
\centering
\begin{tikzpicture}
    % node set with absolute coordinates
    \node[state] (xi) at (0,0) {$\Xi$};

    % Locations can be:
    % right,left,above,below,
    % above left,below right, etc
    \node[state,draw=none]  (blank1) [right =of xi] {};
    \node[state] (w) [right =of blank1] {$W$};
    \node[state] (x) [below =of blank1] {$X$};
    \node[state] (P) [below =of x] {$P$};

    % Directed edge
    \path (xi) edge  (x);
    \path (w) edge (x);
    \path (xi) edge (P);
    \path (w) edge (P);

\end{tikzpicture}
\end{figure}
\newpage
A similar structure arises when the dependence between $\Xi$ and $X$ reflects a common cause $V$, as in Figure \ref{fig lag collider}.  In this case, $X$ is again a collider in the path between $\Xi$ and $W$.

\begin{figure}[h!]\caption{}\label{fig lag collider}
\centering
\begin{tikzpicture}
    % node set with absolute coordinates
    \node[state] (xi) at (0,0) {$\Xi$};

    % Locations can be:
    % right,left,above,below,
    % above left,below right, etc
    \node[state] (x) [right =of xi] {$X$};
    \node[state] (w) [right =of x] {$W$};
    \node [state] (v) [above = of xi] {$V$};

    % Directed edge
    \path (v) edge (xi);
    \path (w) edge (x);
    \path (v) edge (x);

\end{tikzpicture}
\end{figure}

We can also obtain this type of structure if $W$ is a proxy for a latent cost shifter affecting firms' choices of of $X$. Figure \ref{fig proxy collider} illustrates, letting $L$ represent latent cost shifters known (along with $\Xi$) by firms when choosing $X$, with $W$ denoting a proxy for this shifter---e.g., prices or product characteristics of the same firm in other markets.
\begin{figure}[h!]\caption{}\label{fig proxy collider}
\centering
\begin{tikzpicture}
    % node set with absolute coordinates
    \node[state] (xi) at (0,0) {$\Xi$};

    % Locations can be:
    % right,left,above,below,
    % above left,below right, etc
    \node[state] (x) [right =of xi] {$X$};
    \node[state] (l) [right =of x] {$L$};
    \node[state] (w) [below =of l] {$W$};

    % Directed edge
    \path (xi) edge (x);
    \path (l) edge (x);
    \path  (l) edge (w);

\end{tikzpicture}
\end{figure}

  Thus, just as there are cases in which each type of instrument discussed in section \ref{sec IV} remains valid when conditioning on on endogenous characteristics $X$, there are are other important cases in which (\ref{eq cond indep}) will fail. Given $W \Ind \Xi$, the key threat to the conditional independence condition (\ref{eq cond indep}) is a case in which $X$ responds both to  the structural errors $\Xi$ and to the candidate instrument $W$ (or to the latent factor that $W$ proxies). In such situations, identification will require different instruments for prices.
  
In many cases such instruments will be easily constructed under natural timing assumptions.  This is a topic we take up in the final section of this appendix. We also note that  when $X$ is a collider, $W$ provides a candidate instrument for $X$. Thus, the problem here can provide part of its own solution: when instruments for prices and the endogenous components of $X_t$ are available, our results extend immediately by expanding $P_t$ and $W_t$ to include all endogenous variables and all instruments, respectively.

\subsection{Averting Colliders: Sequential  Timing} 

The previous section describes a class of situations in which candidate instruments that would be properly excluded unconditional on $X$ would fail to be properly excluded conditional on $X$. A leading case is that of cost shifters (e.g., input prices) that, along with $\Xi$, partially determine firms' choices of product characteristics $X$.  However, in such cases one may be able to obtain valid instruments by exploiting the (typical) sequential timing of a firm's  decisions. For example, physical characteristics  of new automobiles sold in year $\tau$ will reflect design choices made well in advance---in particular, before the input costs for year-$\tau$ production are fully known. Pricing in year $\tau$, on the other hand, will typically take place after those costs are known. Such timing is common to many markets.
And, as in other contexts, the temporal separation of observable choices can offer an identification strategy.\footnote{Familiar examples in IO include  strategies used by \cite{OlleyPakes96}, \cite*{ACF2006}, and others in the literature on estimation of production functions.} Here, for example, even if product characteristics are chosen in response to demand shocks and expected input costs, current-period  \textit{innovations} to input costs can offer candidate instruments for prices.

To illustrate, we introduce the time superscript $\tau$ to all random variables. Let $M^{\tau}$ denote a vector of time-$\tau$ input prices and suppose that $M^{\tau}$ follows the stochastic process   
\begin{equation}
M^{\tau} = \Phi(M^{\tau-1}) + W^{\tau},    
\end{equation}
where  $\Phi$ is a possibly unknown function and $W^{\tau} \Ind (\Xi^{\tau},X^{\tau},M^{\tau-1})$. Given observability of $(M^{\tau},M^{\tau-1})$ in all markets, each $W^{\tau}$ is identified. 
Now suppose that  $X^{\tau}$ is chosen by firms in period $\tau-1$, whereas  prices for period $\tau$ are chosen  at time $\tau$. The causal graph in Figure \ref{fig timing1} illustrates key features of such a model.\footnote{The figure includes an edge from $X^{\tau}$ to $\Xi^{\tau}$. The presence (or  direction) of such an edge is not important to the argument here.} For clarity, here we include the prices $P^{\tau}$ in the graph.

\begin{figure}[h!]\caption{}
\centering
\begin{tikzpicture}\label{fig timing1}
    % node set with absolute coordinates
    \node[state] (xi) at (0,0) {$\Xi^{\tau-1}$};
    \node[state] (xit1) [below =of xi] {$\Xi^{\tau}$};

    % Locations can be:
    % right,left,above,below,
    % above left,below right, etc
    
    %\node[state, draw=none] (blank) [right =of xit0] {};
    
    \node[state] (mt0) [right =of xi] {$M^{\tau-1}$};
    \node[state] (x) [below =of mt0] {$X^{\tau}$};
    \node[state, draw=none] (blank2) [right =of mt0] {};
    \node[state] (wt1) [right =of mt0] {$W^{\tau}$};
    \node[state] (mt1) [below =of wt1] {$M^{\tau}$};
    \node[state] (p) [below =of x] {$P^{\tau}$};

    % Directed edge

    \path (xi) edge (x);
    \path (mt0) edge (x);
    \path (mt0) edge (mt1);
    \path (wt1) edge (mt1) ;

    \path (x) edge (xit1) ; 
    \path (x) edge (p) ; 
    \path (mt1) edge (p) ;
    \path (xit1) edge (p);
    \path (xi) edge (xit1);
    
\end{tikzpicture}
\end{figure}

Here it is clear that neither the contemporaneous cost shifters $M^{\tau}$ nor the lagged cost shifters $M^{\tau-1}$ can serve as  instruments for prices conditional on $X^{\tau}$: $X^{\tau}$ would be a collider, exactly as in the previous section. However, the period-$\tau$ innovation $W^{\tau}$ can do the job.  Because  $W^{\tau}$ alters period-$\tau$ marginal cost, it is relevant for the determination of $P^{\tau}$, conditional on $X^{\tau}$. And, by the d-separation criterion, we see than $W^{\tau}$ is independent of $\Xi^{\tau}$ conditional on $X^{\tau}$.  Indeed, this becomes particularly transparent if we again focus on the ancestral graph relevant for assessing this conditional independence, as in Figure \ref{fig timing2}. There we see that  $W^{\tau}$ is an example of a ``fully exogenous instrument,'' as discussed in section \ref{sec fully exog}. Ultimately, the innovation $W^{\tau}$ is simply a cost shifter that is independent of all else.  The important insight, however, is that natural timing assumptions can allow such fully independent cost shifters to be constructed from measures like input prices that themselves are not independent of $\Xi^{\tau}$ conditional on $X^{\tau}$.

\begin{figure}[h!]\caption{}
\centering
\begin{tikzpicture}\label{fig timing2}
    % node set with absolute coordinates
    \node[state] (xi) at (0,0) {$\Xi^{\tau-1}$};
    \node[state] (xit1) [below =of xi] {$\Xi^{\tau}$};

    % Locations can be:
    % right,left,above,below,
    % above left,below right, etc
    
    %\node[state, draw=none] (blank) [right =of xit0] {};
    
    \node[state] (mt0) [right =of xi] {$M^{\tau-1}$};
    \node[state] (x) [below =of mt0] {$X^{\tau}$};
    \node[state, draw=none] (blank2) [right =of mt0] {};
    \node[state] (wt1) [right =of mt0] {$W^{\tau}$};

    % Directed edge

    \path (xi) edge (x);
    \path (mt0) edge (x);

    \path (x) edge (xit1) ; 
    \path (xi) edge (xit1);
  
\end{tikzpicture}
\end{figure}

 Similar arguments can allow construction of valid instruments from observed markup shifters (e.g.,  market-level demographics) whose lagged values affect firms' choices of $X$. Indeed, one may simply reinterpret $M^{\tau}$ above as a time-$\tau$ markup shifter.

\newpage
 \bibliographystyle{ecta}
\bibliography{berryhaile}

\end{document}